\documentclass[final,12pt,3p]{elsarticle}

\usepackage{amssymb}

\usepackage{natbib}
\biboptions{sort&compress}

\usepackage{amssymb}
\usepackage{amsmath, amsthm}
\usepackage{commath}
\usepackage{color}
\usepackage{bm}
\usepackage{hyperref}
\usepackage{url}
\usepackage{subfigure}

\usepackage{cleveref}
\usepackage{graphicx}
\usepackage{multirow}
\usepackage[table]{xcolor}

\newfont{\fp}{msbm10 at 11pt}

\newtheorem{proposition}{Proposition}[section]

\usepackage{lipsum}
\makeatletter
\def\ps@pprintTitle{%
 \let\@oddhead\@empty
 \let\@evenhead\@empty
 \def\@oddfoot{}%
 \let\@evenfoot\@oddfoot}
\makeatother
\begin{document}

\begin{frontmatter}

\title{The expectation-maximization algorithm for autoregressive models with normal inverse Gaussian innovations}

\author[label1]{Monika S. Dhull\corref{cor1}}\ead{2018maz0005@iitrpr.ac.in }
\cortext[cor1]{Corresponding author.}
\author[label1]{Arun Kumar}
\author[label2]{Agnieszka Wy\l oma\'nska}

\address[label1]{Department of Mathematics, Indian Institute of Technology Ropar, Rupnagar, Punjab - 140001, India}
\address[label2]{Faculty of Pure and Applied Mathematics, Hugo Steinhaus Center, Wroclaw University of Science and Technology, Wyspianskiego 27, 50-370 Wroclaw, Poland}

\begin{abstract}
The autoregressive (AR) models are used to represent the time-varying random process in which output depends linearly on previous  terms and a stochastic term (the innovation). In the classical version, the AR models are based on normal distribution. However, this distribution does not allow describing data with outliers and asymmetric behavior. In this paper, we study the AR models with normal inverse Gaussian (NIG) innovations. The NIG distribution belongs to the class of semi heavy-tailed distributions with wide range of shapes and thus allows for describing real-life data with possible jumps. The expectation-maximization (EM) algorithm is used to estimate the parameters of the considered model. The efficacy of the estimation procedure is shown on the simulated data. A comparative study is presented, where the classical estimation algorithms are also incorporated, namely, Yule-Walker  and conditional least squares methods along with EM method for model parameters estimation. The applications of the introduced model are demonstrated on the real-life financial data.
\end{abstract}


\end{frontmatter}

\section{Introduction}
In time series modeling, the autoregressive (AR) models are used to represent the time-varying random process in which output depends linearly on previous  terms and a stochastic term also known as error term or innovation term. In the classical approach, the marginal distribution of the innovation terms is assumed to be normal. However, the application of non-Gaussian distributions in time series allows modeling the outliers and asymmetric behavior visible in many real data. In a study conducted by Nelson and Granger \cite{nelson1979}, it was shown that out of 21 real time series data, only 6 were found to be normally distributed. In addition, in financial markets the distribution of the observed time series (log-returns) are mostly non-Gaussian and have tails heavier than the normal distribution, however lighter than the power law. These kind of distributions are also called semi heavy-tailed, see, e.g., \cite{Rachev2003,Cont2004,Omeya2018}. The AR models with non-Gaussian innovations are very well studied in the literature. Sim \cite{Sim1990} considered AR model of order $1$ with Gamma process as the innovation term. For AR models with innovations following a Student's $t$-distribution, see, e.g., \cite{Tiku2000,Tarami2003,Christmas2011,Nduka2018} and references therein. Note that Student's $t$-distribution is used in modeling of asset returns \cite{Heyde2005}.

One of the semi heavy-tailed distributions with wide range of shapes is normal inverse Gaussian (NIG), which was introduced by Barndorff-Nielsen \cite{Barndorff1997a}. NIG distributions were used to model the returns from the financial time-series \cite{aw1, aw2, aw4, aw5, aw6, aw7, aw10}. In practical situations, we come across the data with skewness, extreme values or with missing values which can be easily assimilated by the NIG distribution. The distribution has stochastic representation, i.e., it can be written as the normal variance-mean mixture where the mixing distribution is the inverse Gaussian distribution and a more general distribution known as generalised hyperbolic distribution is obtained by using generalised inverse Gaussian as the mixing distribution \cite{Barndorff2013}. In recent years, new methods dedicated to analyze and estimate the parameters related to NIG distributions based models were introduced which is a testimony of their popularity, see, e.g. \cite{stt1,stt2,stt4,stt5}.

It is worth mentioning, apart from applications in economics and finance, the NIG distributions have found interesting applications in many other areas, such as computer science \cite{aw8}, energy markets \cite{aw11}, commodity markets \cite{aw14} and image analysis \cite{stt7}. The multivariate NIG distributions, counterparts of univariate NIG distributions, were considered  in \cite{stt3,stt6} and were applied in various disciplines, see e.g. \cite{aw13}.

In the literature, there are also various models and stochastic processes that are based on NIG distribution. The very first example is the NIG L\'evy process, see e.g, \cite{ma4} which was also applied to financial data modeling \cite{ma7}. There are also numerous time series models with NIG distributed innovations and their various applications. We mention here the interesting analysis of heteroscedastic models \cite{aaw3,ma1,ma2,ma3,ma5} and autoregressive models \cite{ma8}, see also \cite{ma6}.

In this paper, we study the autoregressive model of order $p$ (an AR($p$)) with NIG innovations. Because the NIG distribution tails are heavier than the normal one, the introduced model can capture large jumps in the real-life time-series data which are quite ubiquitous. Moreover, by applying the autoregressive filter, the considered model can capture the possible short dependence in the data that is characteristic for financial time series returns. We introduce a new estimation algorithm for the analyzed model's parameters. A step-by-step procedure for model parameters' estimation based on expectation-maximization (EM) algorithm is proposed. According to our knowledge, the EM algorithm was not used for the time series models with NIG distributed innovations. For EM-based approach for the NIG distributions, see, e.g. \cite{stt4}. Thus, in this sense we extend this research. The efficacy of the proposed estimation procedure is shown for the simulated data. For the comparison, the model parameters are also estimated using Yule-Walker (YW) and conditional least square (CLS) methods, the most known algorithms for the AR models with finite-variance residuals. The comparative study clearly indicates that the EM-based algorithm outperforms the other considered methods. Finally, the introduced model applications are demonstrated on NASDAQ stock market index data. The introduced model explains very well the NASDAQ index data which can not be modeled using AR model with normally distributed innovations. 

The rest of the paper is organized as follows: In Section 2, first we provide the main properties of the NIG distribution and then the $AR(p)$ model with NIG innovation term is introduced. A step-by-step procedure to estimate the parameters of the introduced model based on EM algorithm is discussed in Section 3. In Section 4, the efficacy of the estimation procedure is proved for simulated data. In this section we also present the comparative study, where the new technique is compared with the YW and CLS algorithm. Finally, the applications to real financial data are demonstrated. The last section concludes the paper.

\setcounter{equation}{0}
\section{NIG autoregressive model}
In this section, we introduce the AR($p$) model having independent identically distributed (i.i.d.) NIG innovations. However, first we remind the definition and main properties of the NIG distribution. \\
\noindent{\bf NIG distribution:} A random variable $X$ is said to have a NIG distribution which is denoted by $X\sim$ NIG $(\alpha, \beta, \mu, \delta),$ if its probability density function (pdf) has the following form
\begin{equation}\label{NIG_pdf}
f(x; \alpha, \beta, \mu, \delta) = \frac{\alpha}{\pi}\exp\left(\delta \sqrt{\alpha^2 - \beta^2} - \beta \mu\right)\phi(x)^{-1/2}K_{1}(\delta \alpha \phi(x)^{1/2})\exp(\beta x),\;x\in \mathbb{R},
\end{equation}
where $\phi(x) = 1+[(x-\mu)/\delta ]^2$, $\alpha = \sqrt{\gamma^2+\beta^2}$, $0\leq |\beta|\leq \alpha,\; \mu\in \mathbb{R},\; \delta>0$ and $K_\nu(x)$ denotes the modified Bessel function of the third kind of order $\nu$ evaluated at $x$ and is defined by
$$
K_\nu(x) = \frac{1}{2}\int_{0}^{\infty}y^{\nu-1}e^{-\frac{1}{2}x(y+y^{-1})}dy.
$$
Using the asymptotic properties of the modified Bessel function, $K_{\nu}(x) \sim \sqrt{\frac{\pi}{2}}e^{-x}x^{-1/2}$ \cite{Jorgensen1982} and the fact that $\phi(x) \sim (x/\delta)^2$ as $x \rightarrow \infty$, we have the following expression
\begin{align*}
f(x; \alpha, \beta, \mu, \delta) &\sim  \frac{\alpha}{\pi} e^{(\delta \sqrt{\alpha^2 - \beta^2} - \beta \mu)}\phi(x)^{-1/2} \sqrt{\frac{\pi}{2}} e^{-\delta\alpha\sqrt{\phi(x)}} \left(\delta\alpha \sqrt{\phi(x)}\right)^{-1/2}e^{\beta x}\\
& \sim \frac{\alpha}{\pi} e^{(\delta \sqrt{\alpha^2 - \beta^2} - \beta \mu)}\sqrt{\frac{\pi}{2}}(\delta\alpha)^{-1/2}\phi(x)^{-3/4}e^{-\delta\alpha\sqrt{\phi(x)}} e^{\beta x}\\
&\sim  \frac{\alpha}{\pi} e^{(\delta \sqrt{\alpha^2 - \beta^2} - \beta \mu)}\sqrt{\frac{\pi}{2}}(\delta\alpha)^{-1/2} \delta^{3/4}x^{-3/4}e^{\mu\alpha} e^{-(\alpha-\beta)x},\;\alpha>\beta,\\
&\sim  \sqrt{\frac{\alpha}{2\pi}} e^{(\delta \sqrt{\alpha^2 - \beta^2} - \beta \mu)}\delta x^{-3/2} e^{-(\alpha-\beta)x},\;\alpha>\beta,\;\; {\rm as}\;\; x\rightarrow \infty.
\end{align*}

\noindent From the above, one can conclude that the tail probability for NIG distributed random variable $X$ satisfies the following
$$\mathbb{P}(X>x) \sim c x^{-3/2} e^{-(\alpha-\beta)x},\;\; {\rm as}\;\; x\rightarrow \infty,$$ 
where $c=\sqrt{\frac{\alpha}{2\pi}}\frac{\delta}{(\alpha-\beta)} e^{(\delta \sqrt{\alpha^2 - \beta^2} - \beta \mu)}$, which shows that NIG is a semi-heavy tailed distribution \cite{Rachev2003,Cont2004,Omeya2018}.
It is worth mentioning that the NIG distributed random variable $X$ can be represented in the following form
\begin{align}\label{mean_variance_mixture}
X = \mu + \beta G + \sqrt{G} Z,
\end{align} where $Z$ is a standard normal random variable i.e. $Z\sim N(0,1)$  and $G$ has an inverse Gaussian (IG) distribution  with parameters $\gamma$ and $\delta$ denoted by $G\sim$ IG$(\gamma, \delta)$, having pdf of the following form
\begin{equation}\label{IG_density}
g(x; \gamma, \delta) = \frac{\delta}{\sqrt{2\pi}}\exp(\delta \gamma)x^{-3/2}\exp\left(-\frac{1}{2}\left(\frac{\delta^2}{x} + \gamma^2x\right)\right),\;x>0.
\end{equation}
The representation given in Eq. (\ref{mean_variance_mixture}) is useful when we generate the NIG distributed random numbers. It is also suitable to apply the EM algorithm for the maximum likelihood (ML) estimation of the considered model's parameters. 
The representation (\ref{mean_variance_mixture}) makes it convenient to find the main characteristics of a NIG distributed random variable $X$, namely, we have
$$
\mathbb{E}X = \mu + \delta\frac{\beta}{\gamma}\;\;\;\;\mathrm{and}\;\;\;\;\; \mathrm{Var}(X) = \delta \frac{\alpha^2}{\gamma^3}.
$$
Moreover, the skewness and kurtosis of $X$ is given by
$$
\rm{Skewness} = \frac{3\beta}{\alpha\sqrt{\delta\gamma}}\;\;\; Kurtosis = \frac{3(1+4\beta^2/\alpha^2)}{\delta\gamma}.
$$
For $\beta =0$, the NIG distribution is symmetric. Moreover, it is leptokurtic if $\delta\gamma<1$ while it is platykurtic in case $\delta\gamma>1$. Note that a leptokurtic NIG distribution is characterized by larger number of outliers than we have for normal distribution  and thus, it is a common tool for financial data description. 

\noindent{\bf NIG autoregressive model of order $p$:} Now, we can define the AR($p$) univariate stationary time-series $\{Y_t\}$, $t\in \mathbb{Z}$ with NIG innovations
\begin{equation}\label{main_model}
Y_t=\sum_{i=1}^p \rho_i Y_{t-i} + \varepsilon_t = \bm{\rho^{T}}\mathbf{Y_{t-1}} + \varepsilon_{t},
\end{equation}
where $\bm{\rho} = (\rho_1, \rho_2, \cdots, \rho_p)^T$ is a $p$-dimensional column vector, $\mathbf{Y_{t-1}} = (Y_{t-1}, Y_{t-2}, \cdots, Y_{t-p})^T$ is a vector of $p$ lag terms and $\{\epsilon_{t}\}$, $t\in \mathbb{Z}$ are i.i.d. innovations distributed as NIG$(\alpha, \beta, \mu, \delta)$. The process $\{Y_t\}$ is a stationary one if and only if the modulus of all the roots of the characteristic
polynomial $(1-\rho_1z-\rho_2z^2-\cdots-\rho_pz^p)$ are greater than one. In this article, we assume that the error term follows a symmetric NIG distribution with mean 0 i.e. $\mu=\beta=0.$
Using properties of NIG distribution (see \ref{App_A}), the conditional distribution of $Y_t$ given $\bm{\rho}, \alpha, \beta, \mu, \delta $ and the preceding data $\mathcal{F}_{t-1} = (Y_{t-1}, Y_{t-2}, \cdots, Y_1)^T$ is given by
\begin{align*}
p(Y_{t}|\bm{\rho}, \alpha, \beta, \mu, \delta, \mathcal{F}_{t-1})
 = f(y_t; \alpha, \beta, \mu + \bm{\rho^{T}}\mathbf{y_{t-1}}, \delta),
\end{align*}
where $f(\cdot)$ is the pdf given in Eq. \eqref{NIG_pdf} and $\mathbf{y_{t-1}}$ is the realization of $\mathbf{Y_{t-1}}$. We have $\mathbb{E}[Y_t] = \mathbb{E}[\varepsilon_t] =0$ and Var$[Y_t] = \sigma_{\varepsilon}^2 + \sum_{j=1}^{p}\rho_j\gamma_j$, where $\sigma_{\varepsilon}^2 $ = Var$(\varepsilon_t) = \delta\alpha^2/\gamma^3$ and $\gamma_j = \mathbb{E}[Y_tY_{t-j}] = \rho_1\gamma_{j-1} + \rho_2\gamma_{j-2}+\cdots+ \rho_p\gamma_{j-p},\;j\geq 1$ (see \ref{App_B}).
 
\section{Parameter estimation using EM algorithm}
In this section, we provide a step-by-step procedure to estimate the parameters of the model  proposed in Eq. \eqref{main_model}. The procedure is based on EM algorithm. In this paper, we provide estimates of all parameters of the introduced AR($p$) with NIG distributed innovations. It is worth to mention that EM is a general iterative algorithm for model parameter estimation by maximizing the likelihood function in the presence of missing or hidden data. The EM algorithm was introduced in \cite{Dempster1977} and it is considered as an alternative to numerical optimization of the likelihood function. It is popularly used in estimating the parameters of Gaussian mixture models (GMMs), estimating hidden Markov models (HMMs) and model-based data clustering algorithms. Some extensions of EM include the expectation conditional maximization (ECM) algorithm \cite{Meng1993} and expectation conditional maximization either (ECME) algorithm \cite{Liu1995}. For a detailed discussion on the theory of EM algorithm and its extensions we refer the readers to \cite{McLachlan2007}. The EM algorithm iterates between two steps, namely the expectation step ({\it E-step}) and the maximization step ({\it M-step}). In our case, the observed data $X$ is assumed to be from NIG$(\alpha, \beta, \mu, \delta)$ and the unobserved data $G$ follows IG$(\sqrt{\alpha^2-\beta^2}, \delta)$. The {\it E-step} computes the expectation of the complete data log-likelihood with respect to the conditional distribution of the unobserved or hidden data, given the observations and the current estimates of the parameters. Further, in the {\it M-step}, a new estimate for the parameters is computed which maximize the complete data log-likelihood computed in the {\it E-step}. We find the conditional expectation of log-likelihood of complete data $(X, G)$ with respect to the conditional distribution of $G$ given $X$. For $\theta = (\alpha, \beta, \mu, \delta, \bm{\rho}^T)$ we find
$$Q(\theta|\theta^{(k)}) = \mathbb{E}_{G|X,\theta^{(k)}}[\log f(X,G|\theta)|X, \theta^{(k)}],$$
in the {\it E-step} where $\theta^{(k)}$ represents the estimates of the parameter vector at $k$-th iteration. 
Further, in the {\it M-step}, we compute the parameters by maximizing the expected log-likelihood of complete data found in the {\it E-step} such that
$$\theta^{(k+1)} = \operatorname*{argmax}_\theta Q(\theta|\theta^{(k)}).$$

The algorithm is proven to be numerically stable \cite{McLachlan2007}. Also, as a consequence of Jensen's inequality, log-likelihood function at the updated parameters $\theta^{(k+1)}$ will not be less than that at the current values $\theta^{(k)}$. Although there is always a concern that the algorithm might get stuck at local extrema, but it can be handled by starting from different initial values and comparing the solutions. In next proposition, we provide the estimates of the parameters of the model defined in eq. \eqref{main_model} using EM algorithm.

\begin{proposition}\label{NIG_AR(p)} 
Consider the AR($p$) time-series model given in Eq. \eqref{main_model} where error terms follow NIG$(\alpha, \beta, \mu,\delta)$. The maximum likelihood estimates of the model parameters using EM algorithm are as follows
\begin{align}
\bm{\hat{\rho}} = \left(\displaystyle\sum_{t=1}^{n}w_{t} \mathbf{Y_{t}Y_{t-1}^{T}}\right)^{-1}  \sum_{t=1}^{n}\left(w_{t} y_{t}-\mu w_{t} - \beta\right)\mathbf{Y_{t-1}},
\end{align}

\begin{align}\label{EM_estimates}
\begin{split}
\hat {\mu} &= \frac{\displaystyle\sum_{t=1}^{n} \epsilon_{t}w_{t} - n \beta}{n \bar{w_{t}}}, \\
\hat {\beta} &= \frac{ \displaystyle\sum_{t=1}^{n}(w_{t} \epsilon_{t}) - n\bar{{w_t}} \bar{\epsilon_{t}}}{n(1 - \bar{s_t}\bar{w_t})},\\
\hat {\delta} &= \sqrt{\frac{\bar{s}}{(\bar{s} \bar{w} - 1)}},\;\; \hat {\gamma} = \frac{\delta}{\bar{s}},\;\mbox{and}\;
\hat{\alpha} = (\gamma^2 + \beta^2)^{1/2},
\end{split}
\end{align} 
where $\epsilon_t = y_t - \bm{\rho^{T}}\mathbf{Y_{t-1}}$, $\bar{\epsilon}_t = \frac{1}{n}\sum_{t=1}^{n}\epsilon_t$, $s_t = \mathbb{E}_{G|\varepsilon,\theta^{(k)}}(g_t|\epsilon_t, \theta^{(k)})$, $\bar{s} = \frac{1}{n}\sum_{t=1}^{n}s_t$,  $w_t = \mathbb{E}_{G|\varepsilon,\theta^{(k)}}(g_t^{-1}|\epsilon_t, \theta^{(k)})$
 and $\bar{w} = \frac{1}{n}{\sum_{t=1}^{n}w_t}. $
\end{proposition}
\begin{proof}
See \ref{App_C}.
\end{proof}

\section{Simulation study and applications}
In this section, we illustrate the performance of the proposed model and the introduced estimation technique using simulated data sets and real time series of NASDAQ stock exchange data.

\subsection{Simulation study}
We discuss the estimation procedure for AR(2) and AR(1) models. The model \eqref{main_model}  is simulated in two steps. In the first step, the NIG innovations are simulated using the normal variance-mean mixture form \eqref{mean_variance_mixture}. For NIG random numbers, standard normal and IG random numbers are required. The algorithm mentioned in \cite{Devroye1986} is used to generate iid IG distributed random numbers $G_i \sim IG(\mu_1, \lambda_1),\;i=1,2,\cdots,N$ using the following steps:
\begin{itemize}
    \item []{\it Step 1}: Generate standard normal variate $Z$ and set $Y = Z^2$.
    \item[] {\it Step 2}: Set $X_1 = \mu_1 + \frac{\mu_{1}^2 Y}{ 2 \lambda_1} - \frac{\mu_1 }{2 \lambda_1}\sqrt{4\mu_1 \lambda_1 Y + \mu_1^2 Y^2}$.
    \item []{\it Step 3}: Generate  uniform random variate $U[0,1]$.
    \item[]{\it Step 4}: If $U<= \frac{\mu_1}{\mu_1 + X_1}$, then $G = X_1$; else $G = \frac{\mu_1^{2}}{X_1}.$ 
\end{itemize}

\noindent Note that the substitutions for parameters as $\mu_1 = \delta/\gamma$ and $\lambda_1 = \delta^2$ are required in the above algorithm because the pdf taken in \cite{Devroye1986} is different from the form given in \eqref{IG_density}. Again we simulate a standard normal vector of size $N$ and use \eqref{mean_variance_mixture} with simulated IG random numbers to obtain the NIG random numbers of size $N$. In step 2, the simulated NIG innovations and the relation given in \eqref{main_model} are used to generate AR($p$) simulated series.

\noindent \textbf{Case 1:} In the simulation study, first we analyze the AR(2) model with NIG innovations. In the analysis we used  $1000$ trajectories of length $N = 1000$ each. The used parameters of the model are: $\rho_1 = 0.5$ and $\rho_2 = 0.3$ while the residuals were generated from NIG distribution with $\alpha=1,\; \beta=0,\; \mu=0$ and $\delta=2$. The exemplary time series data plot and scatter plot of innovation terms  are shown in Fig. \ref{fig1}. 
\begin{figure}[ht!]
\centering
\subfigure[The time series data plot.]{
\includegraphics[width=8cm, height=5.5cm]{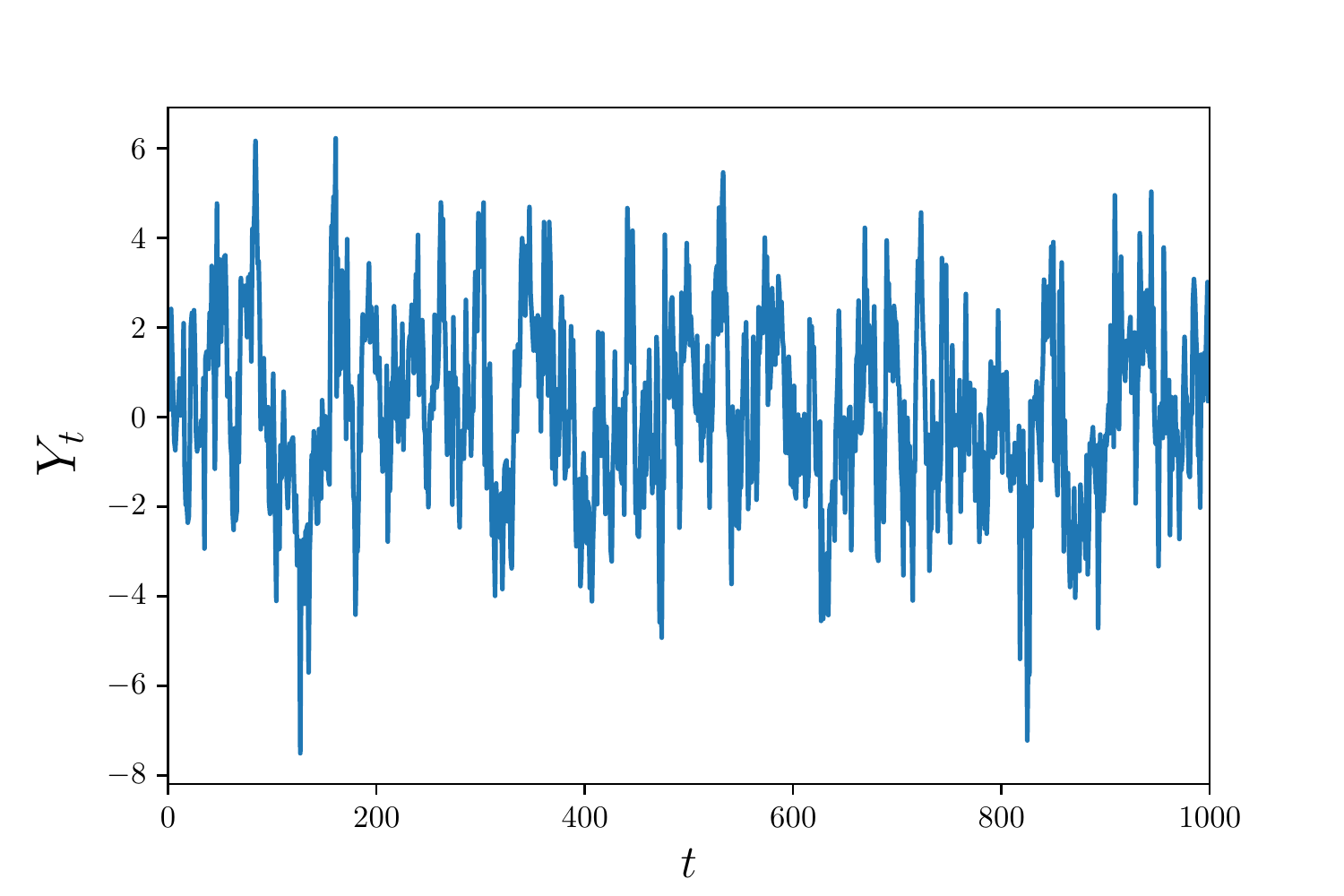}}
\subfigure[Scatter plot of innovation term.]{
\includegraphics[width=8cm, height=5.5cm]{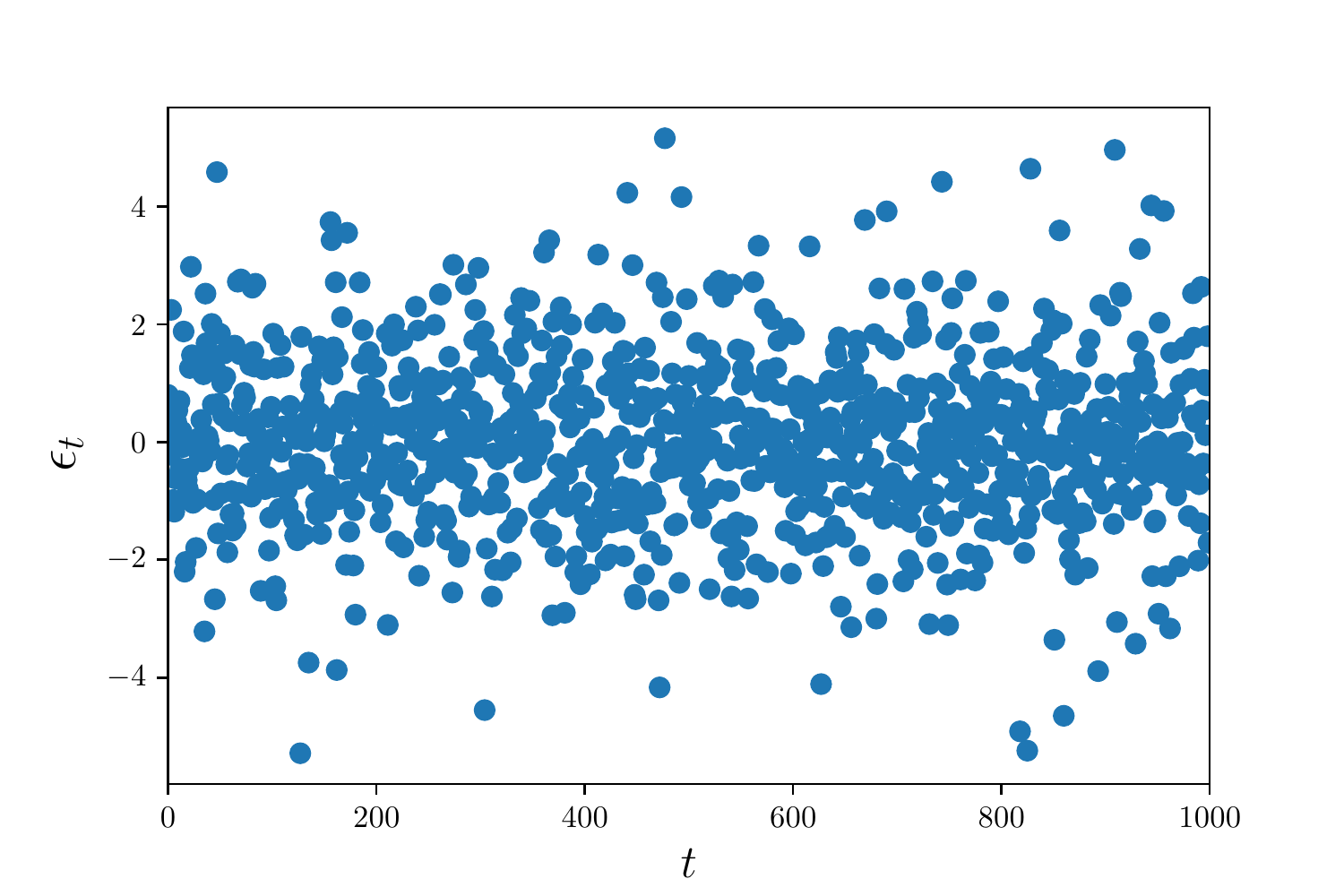}}
\caption{The exemplary time series of length $N=1000$ (left panel) and the corresponding innovation term (right panel) of the AR(2) model with NIG distribution. The parameters of the model are: $\rho_1 = 0.5, \rho_2 = 0.3, \alpha=1, \beta=0, \mu=0, \text{ and } \delta=2$.}\label{fig1}
\end{figure} 
Now, for each simulated trajectory, we apply the estimation algorithm presented in the previous section. For EM algorithm several stopping criteria could be used. One of the examples is the criterion based on change in the log-likelihood function which utilizes the relative change in the parameters' values. We terminate the algorithm when the following criterion for the relative change in the parameters' values is satisfied (this criterion is commonly used in literature)
\begin{equation}\label{stopping_criterion}
\max \left\{\abs{\frac{\alpha^{(k+1)} - \alpha^{(k)}}{\alpha^{(k)}}}, \abs{\frac{\delta^{(k+1)} - \delta^{(k)}}{\delta^{(k)}}}, \abs{\frac{\rho^{(k+1)} - \rho^{(k)}}{\rho^{(k)}}} \right\} < 10^{-4}.
\end{equation}
The parameters' estimates obtained from the simulated data are shown in the boxplot in Fig. \ref{fig2}. Moreover, we compared the estimation results with the classical YW algorithm and CLS method. We remind that the YW algorithm is based on the YW equations calculated for the considered model, and utilizes the empirical autocovariance function for the analyzed data. More details of YW algorithm for autoregressive models can be found, for instance in \cite{brockwell2016introduction}.  
{We remind, the CLS method estimates the model parameters for dependent observations by minimizing the sum of squares of deviations about the conditional expectation.

Fig. \ref{fig2}(a) and Fig. \ref{fig2}(b) represent the estimates of the model parameters $\rho_1$ and $\rho_2$ using YW, CLS and EM methods, respectively. Furthermore, using the estimated $\rho_1$ and $\rho_2$ parameters with YW and CLS methods the residuals or innovation terms are obtained and then again EM algorithm is used to estimate the remaining $\alpha$ and $\delta$ parameters which are plotted in Fig. \ref{fig2}(c) and \ref{fig2}(d). Moreover, the estimates for $\alpha$ and $\delta$ using directly EM algorithm given in \eqref{EM_estimates} are also plotted in Fig. \ref{fig2}(c) and \ref{fig2}(d). From boxplots presented in Fig. \ref{fig2} we observe that the estimates of $\rho_1$ and $\rho_2$ parameters using the EM algorithm have less variance in comparison to the YW and CLS algorithms. Moreover, for $\alpha$ and $\delta$ parameters, we see that the means of the estimates for the three presented methods are close to the true values, but the range of outliers for the EM algorithm is comparatively less. Therefore, we can infer that EM algorithm performs better (in comparison of other considered algorithms) in the parameter estimation of AR($p$) model with NIG innovations.}

\begin{figure}[ht!]
\centering
\subfigure[Boxplots for $\rho_{1}$ estimates.]{
\includegraphics[width=8cm, height=5cm]{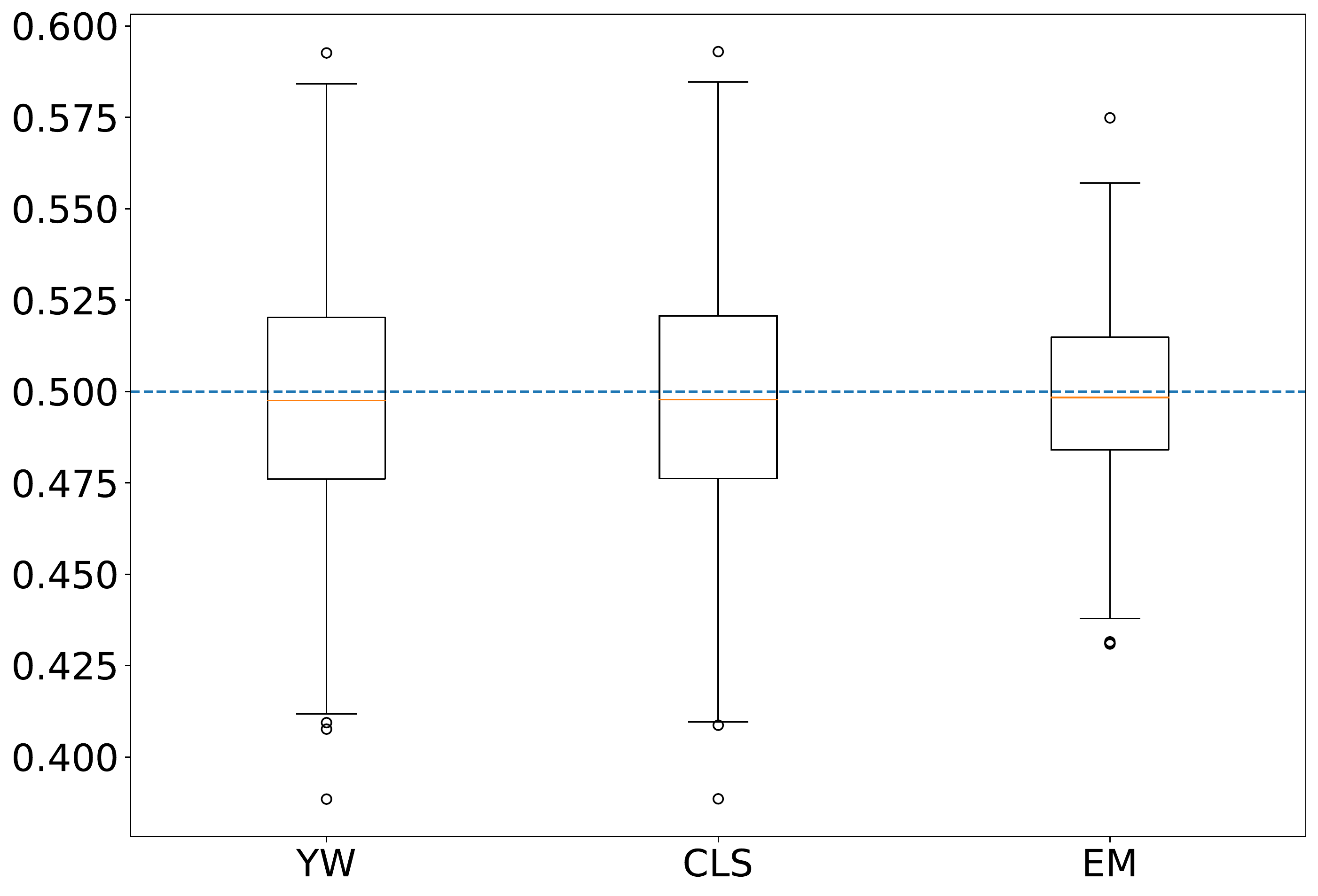}}
\subfigure[Boxplots for $\rho_{2}$ estimates.]{
\includegraphics[width=8cm, height=5cm]{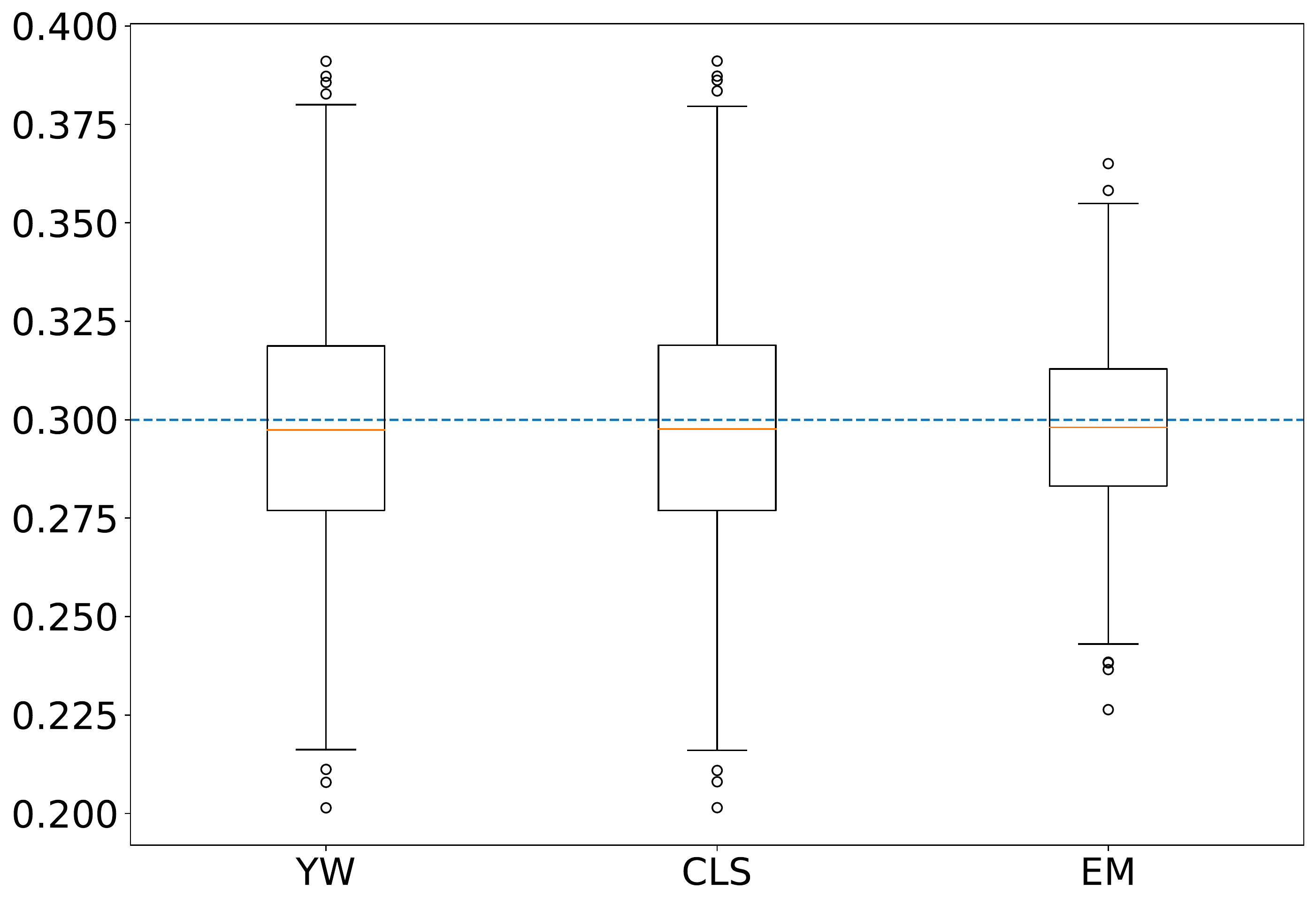}}
\subfigure[Boxplots for $\delta$ estimates.]{
\includegraphics[width=8cm, height=5cm]{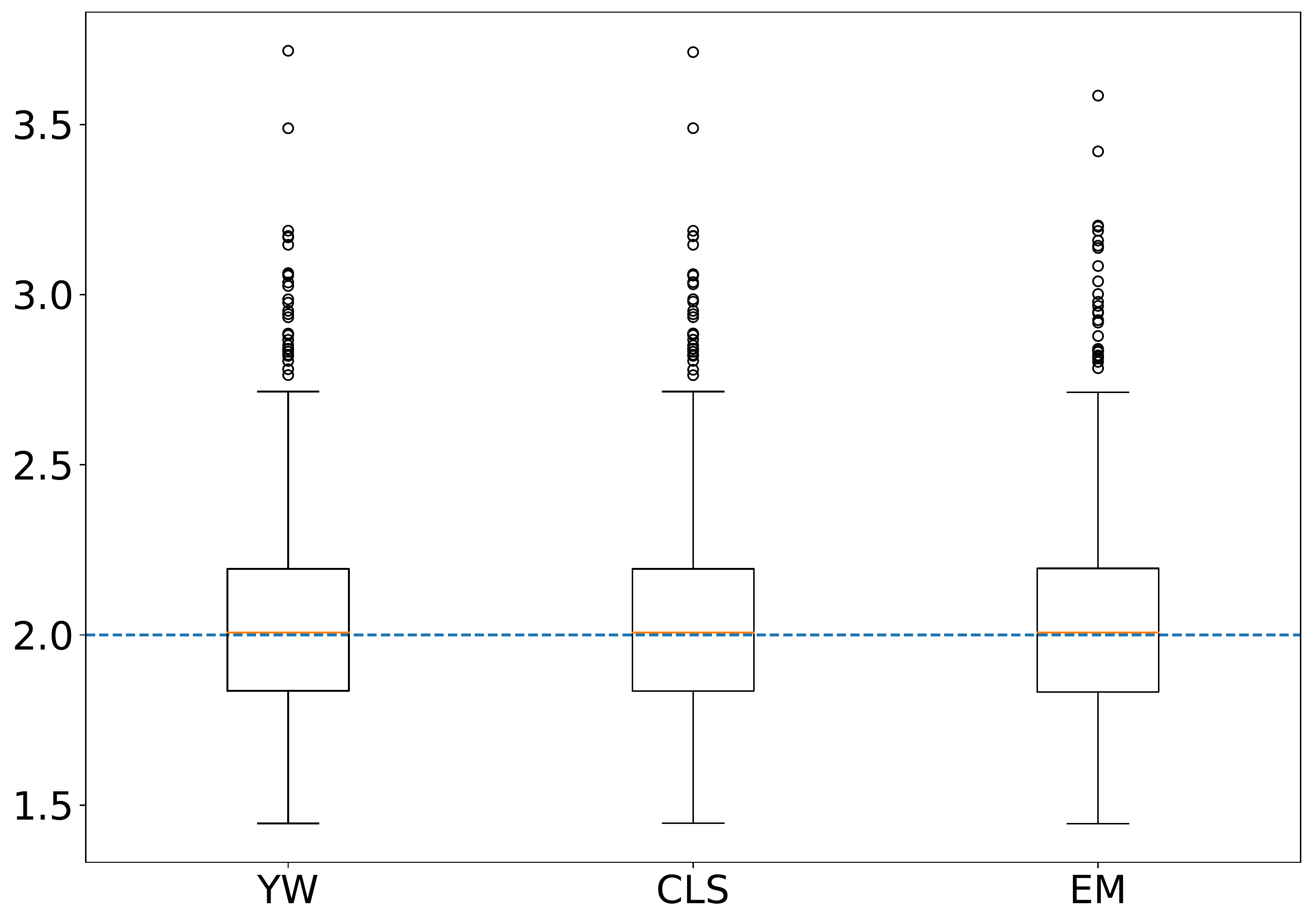}}
\subfigure[Boxplots for $\alpha$ estimates.]{
\includegraphics[width=8cm, height=5cm]{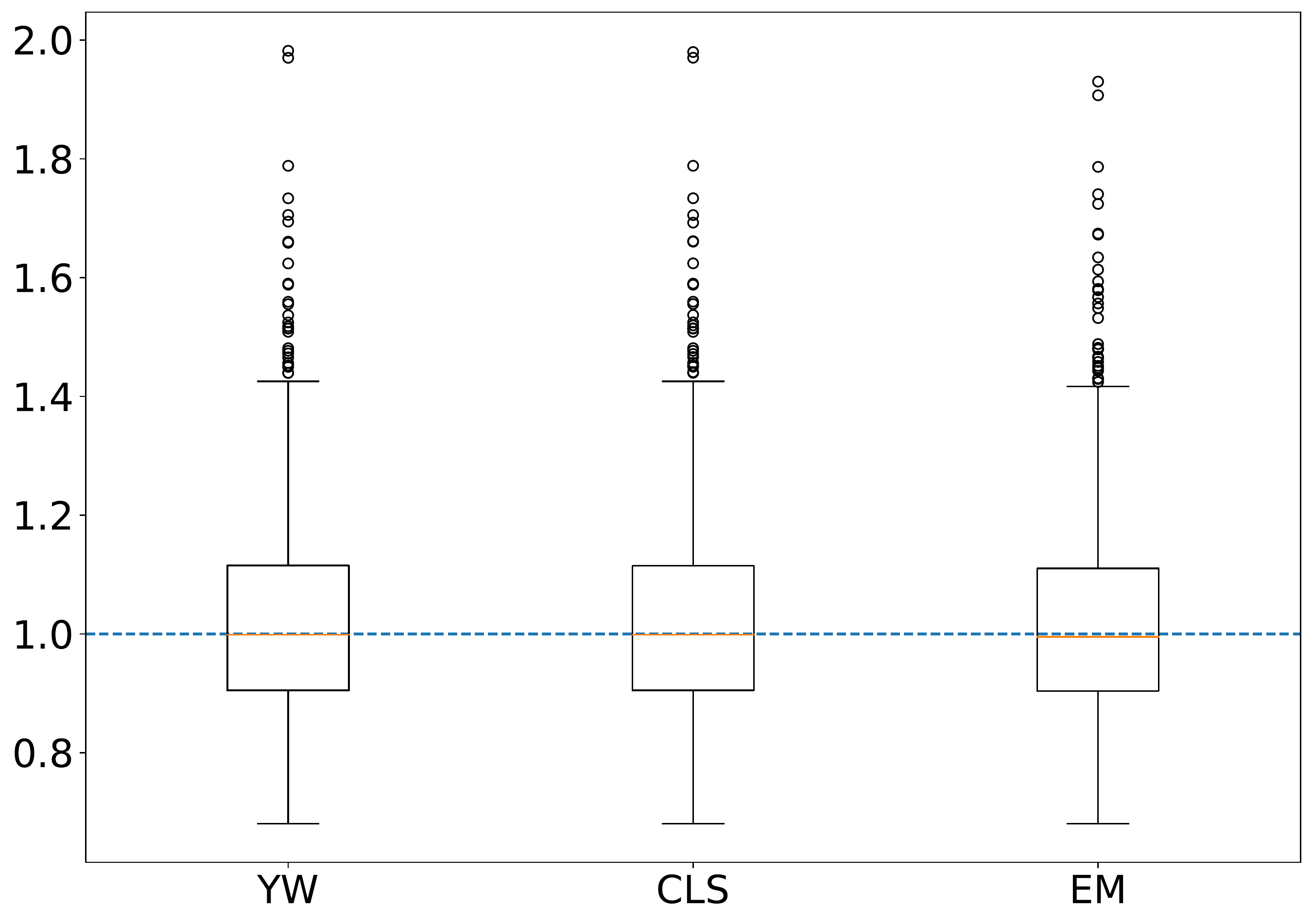}}
\caption{Boxplots of the estimates of the AR(2) model's parameters with theoretical values: $\rho_1 = 0.5, \rho_2 = 0.3, \delta = 2 \text{ and } \alpha = 1$ represented with blue dotted lines. The boxplots are created using $1000$ trajectories each of length $1000$.}\label{fig2}
\end{figure}
  
\noindent \textbf{Case 2:} As the second example, we analyze the AR(1) model with NIG innovations. Here we examine the trajectories of $579$ data points. The same number of data points are examined in the real data analysis demonstrated in the next subsection. This exemplary model is discussed to verify the results for the real data. The simulated errors follow from NIG distribution with parameter $\alpha=0.0087,\; \beta=0,\; \mu=0$ and  $\delta=70.3882$ while the model's parameter is $\rho=0.9610$. In Fig. \ref{fig3}, we present the exemplary simulated trajectory and the corresponding innovation terms.
\begin{figure}[ht!]\centering
\subfigure[Time series data plot]{
\includegraphics[width=8cm, height=5.5cm]{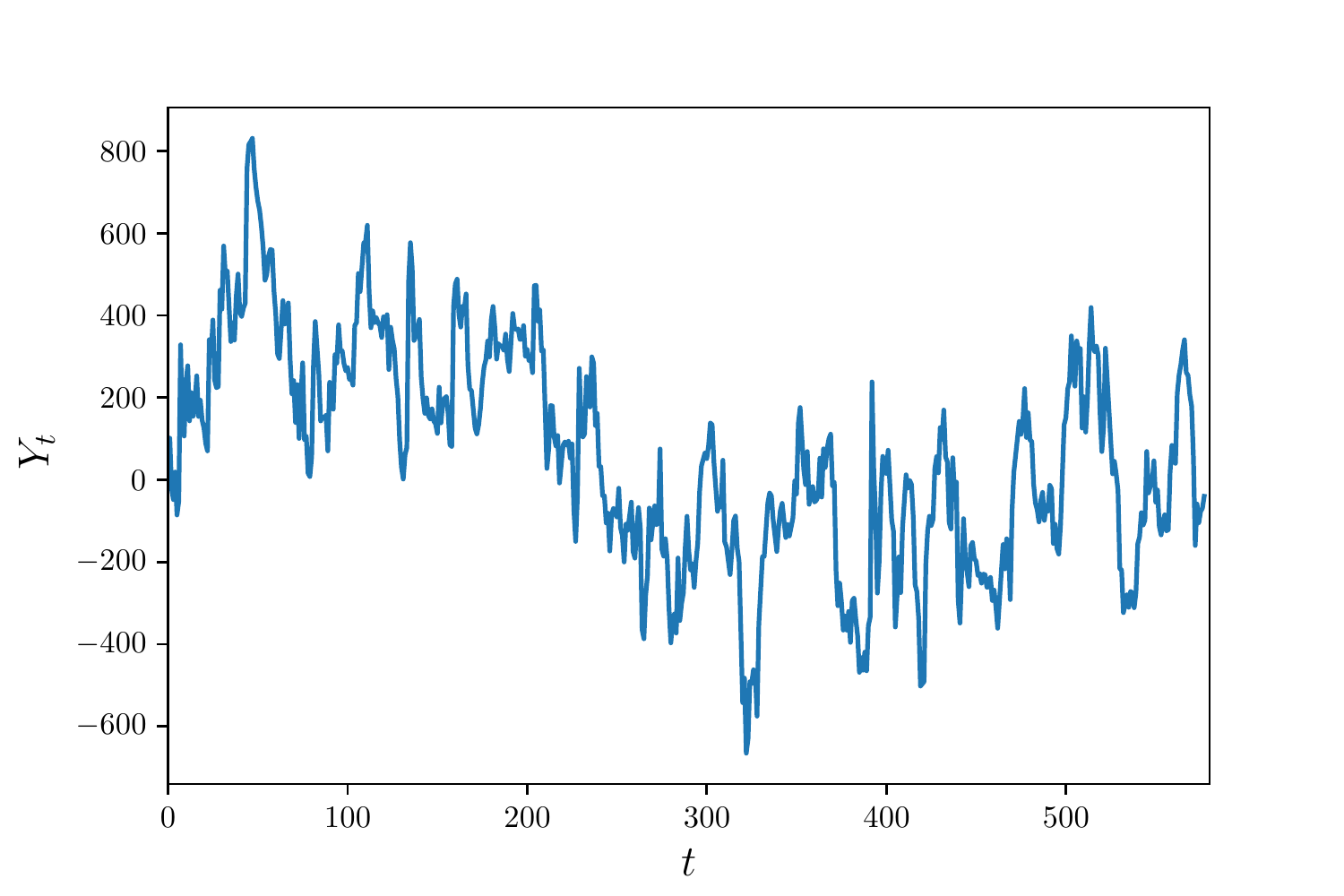}}
\subfigure[Scatter plot of innovation terms]{
\includegraphics[width=8cm, height=5.5cm]{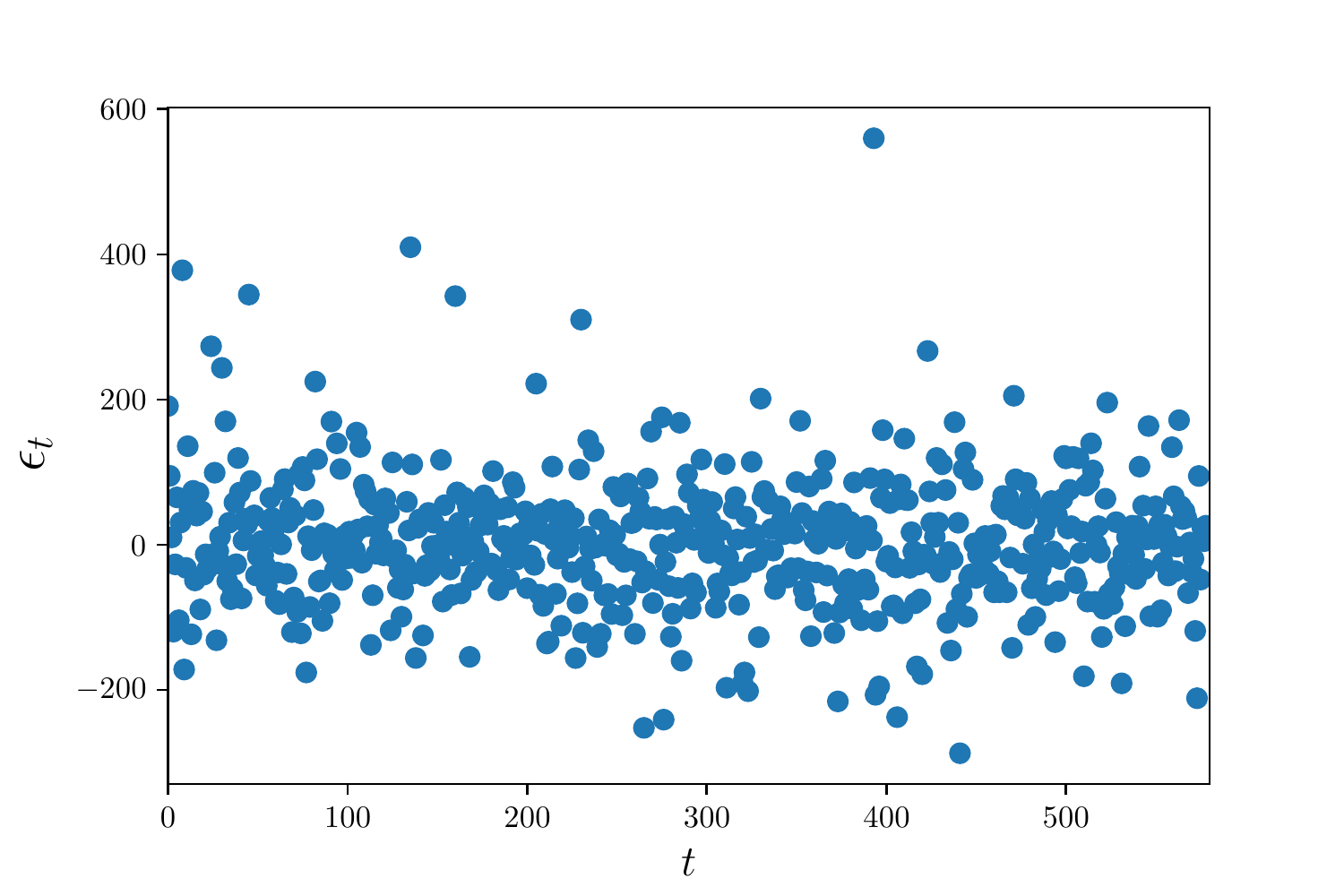}}
\caption{The exemplary time series plot of the first trajectory of length $N=579$ from AR(1) model (left panel) with the corresponding scatter plot of innovation terms NIG distribution (right panel). The parameters of the model are $\rho = 0.961,\alpha=0.0087, \beta=0, \mu=0 \text{ and }\delta=70.3882$.}\label{fig3}  
\end{figure} 
Similarly as in Case 1, the introduced EM algorithm was applied to the simulated data with the stopping criteria based on the relative change in the parameter values defined in Eq. \eqref{stopping_criterion}. 
The boxplots of the estimated parameters for $1000$ trajectories each of length $579$ are shown in Fig. \ref{fig4}. Similar as in the previous case, we compare the results for EM, YW and CLS algorithms. 

\begin{figure}[ht!]
\centering
\subfigure[Boxplot for $\rho$ estimate.]{
\includegraphics[width=8cm, height=5.5cm]{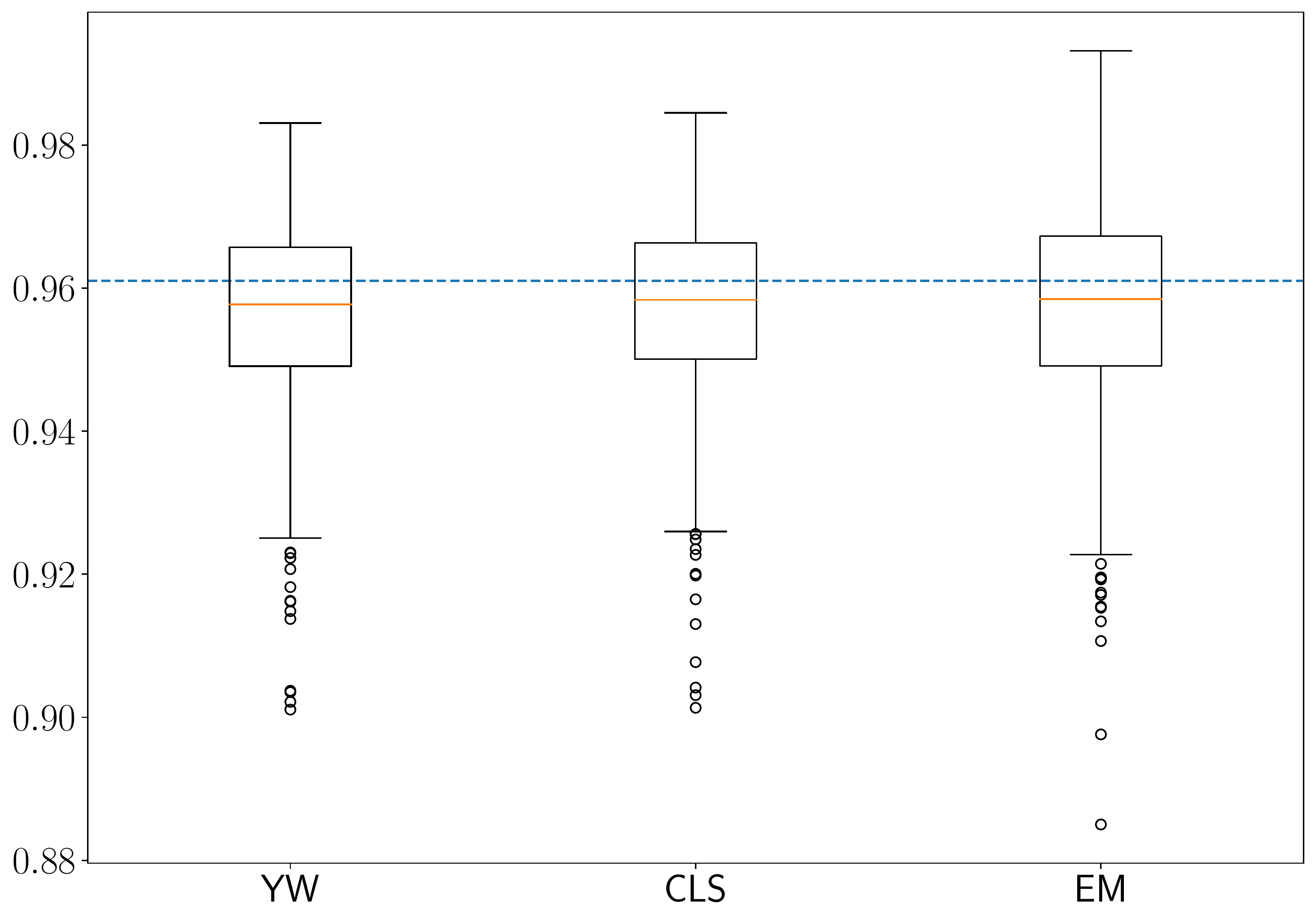}}
\subfigure[Boxplot for $\delta$ estimate.]{
\includegraphics[width=8cm,height=5.5cm]{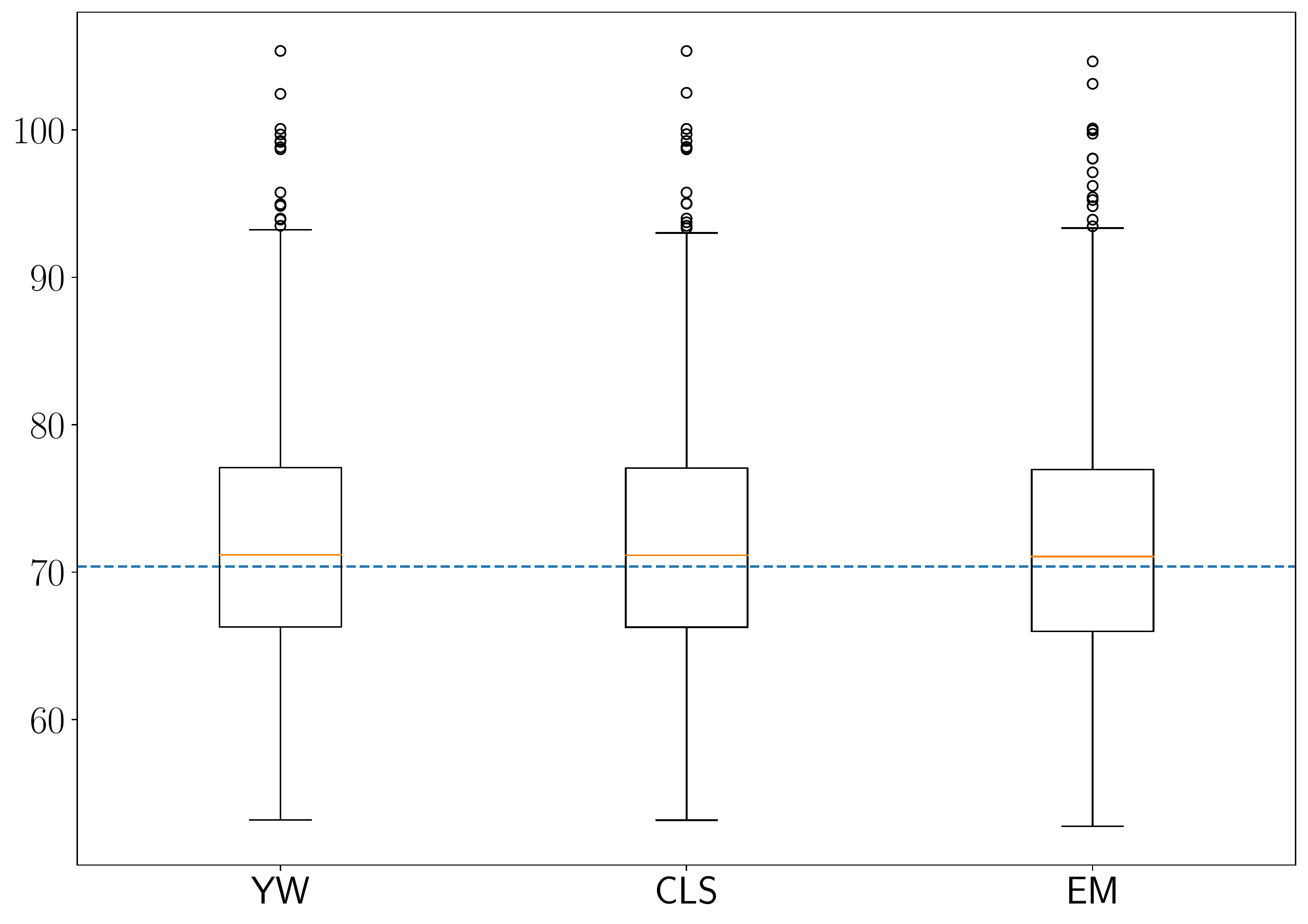}}
\subfigure[Boxplot for $\alpha$ estimate.]{
\includegraphics[width=8cm, height=5.5cm]{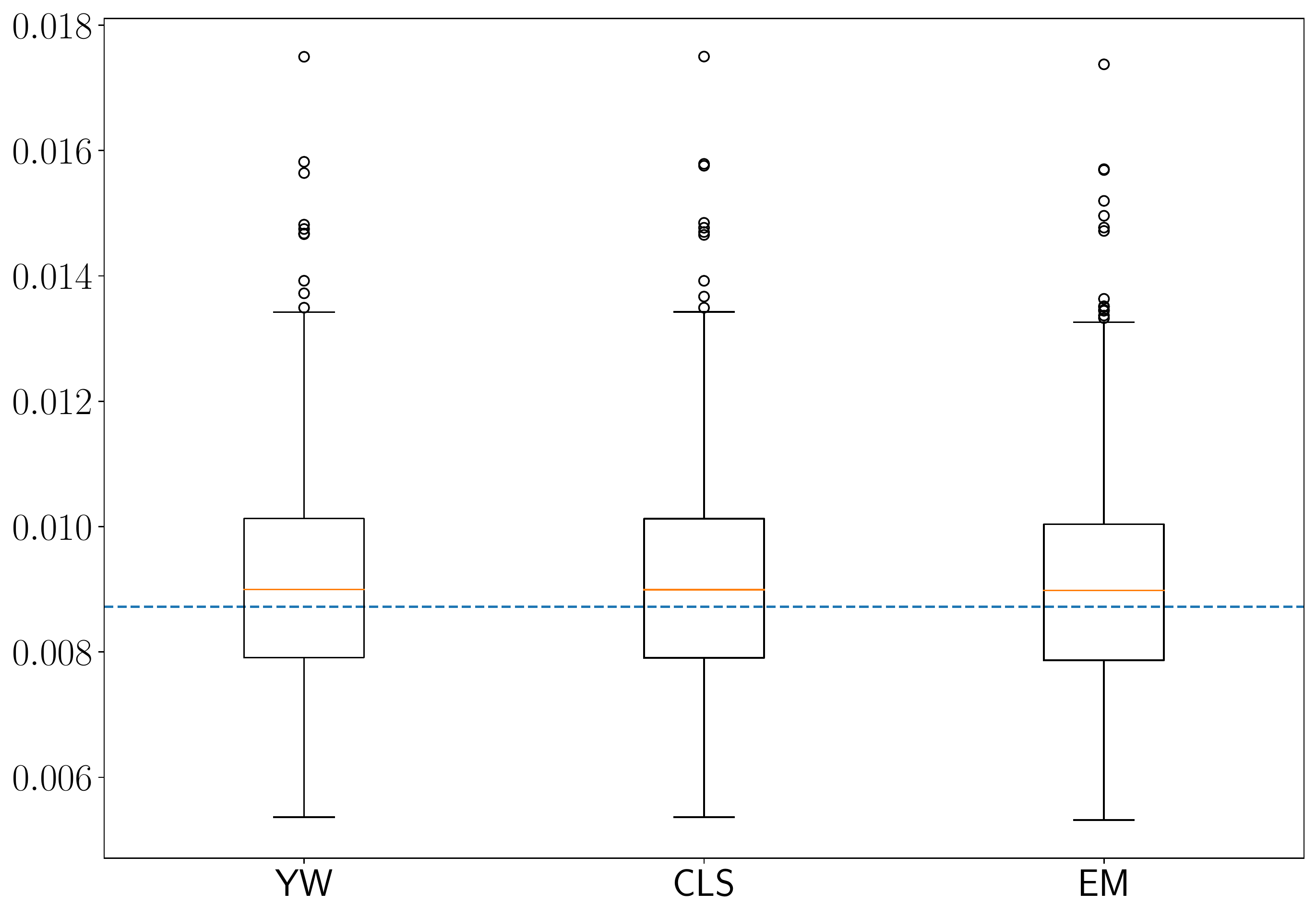}}
\caption{Boxplots of the estimates of the AR(1) model's parameters with theoretical values: $\rho = 0.9610, \delta = 70.3883 \text{ and } \alpha = 0.00872$ represented with blue dotted lines. The boxplots are created using $1000$ trajectories each of length $579$.}\label{fig4}
\end{figure}
{From Fig. \ref{fig4} one can observe that although the estimate of $\rho$ has more variance compared to YW and CLS methods, but the estimates $\delta$ and $\alpha$ have less variance and the spread of outliers is also slightly less.
The means of the estimated parameters from $1000$ trajectories of length $N=579$ using EM algorithm are $\hat{\rho} = 0.9572, \; \hat{\delta} = 71.8647 \text{ and } \hat{\alpha} = 0.0091$. We can conclude that the EM algorithm, also in this case, gives the better parameters' estimates for the considered model.}

\subsection{Real data applications}
In this part, the considered AR($p$) model with NIG distribution is applied to the NASDAQ stock market index data, which is available on Yahoo finance \cite{nasdaq}. It covers the historical prices and volume of all stocks listed on NASDAQ stock exchange from the period March 04, 2010 to March 03, 2020. The data consists of $2517$ data points with features having {\it open price}, {\it closing price}, {\it highest value}, {\it lowest value}, {\it adjusted closing price} and {\it volume} of stocks for each working day end-of-the-day values. We choose the end-of-the-day adjusted closing price as a univariate time series for the analysis purpose. The innovation terms  of time series data is assumed to follow NIG distribution as the general one. 
In Fig. \ref{fig5} we represent the adjusted closing price of NASDAQ index. 
\begin{figure}[ht!]\centering
\includegraphics[width=8.5cm, height=6cm]{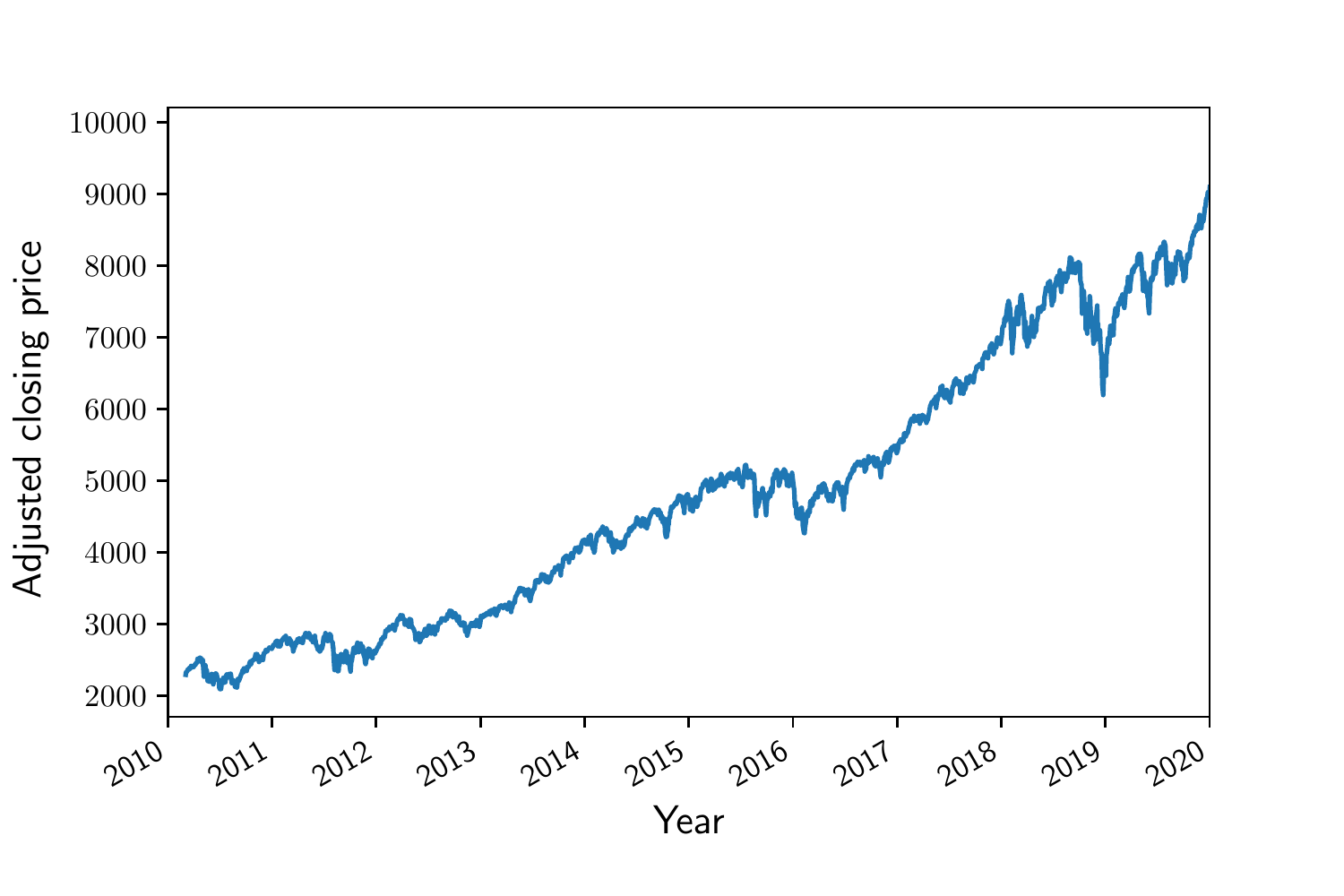}
\caption{The adjusted closing price (in\$) of NASDAQ index from the period March 04, 2010 to March 03, 2020 with $2517$ data points.}\label{fig5}
\end{figure}
Observe that the original time series data has an increasing trend. Moreover, one can easily observe that the data exhibit non-homogeneous behavior. Thus, before further analysis the analyzed time series should be segmented in order to obtain the homogeneous parts. To divide the vector of observations into homogeneous parts, we applied the segmentation algorithm presented in \cite{acta}, where authors proposed to use the statistics defined as the cumulative sum of squares of the data. Finally, the segmentation algorithm is based on the specific behavior of the used statistics when the structure change point exists in the analyzed time series. More precisely, in \cite{acta} it was shown that the cumulative sum of squares is a piece-wise linear function when the variance of the data changes. Because in the considered time series we observe the non-stationary behavior resulting from the existence of the deterministic trend, thus, to find the structure break point, we applied the segmentation algorithm for their logarithmic returns. Finally, the algorithm indicates that the data needs to be divided into two segments, the first $1937$ observations are considered as data 1 and rest all observations as - data 2.
\begin{figure}[ht!]
\centering
\subfigure[Data 1]{
\includegraphics[width=8cm, height=5.5cm]{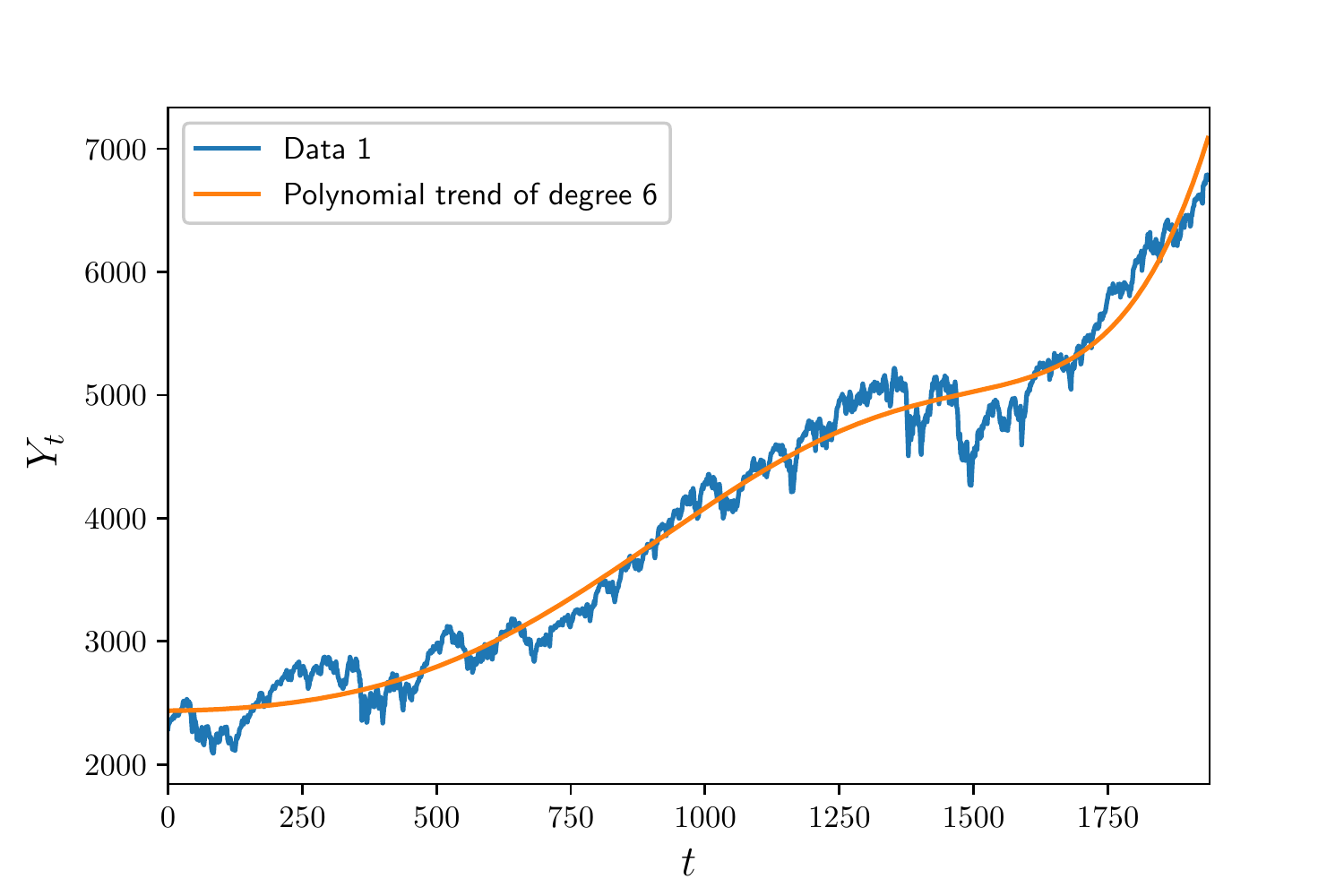}}
\subfigure[Data 2]{
\includegraphics[width=8cm, height=5.5cm]{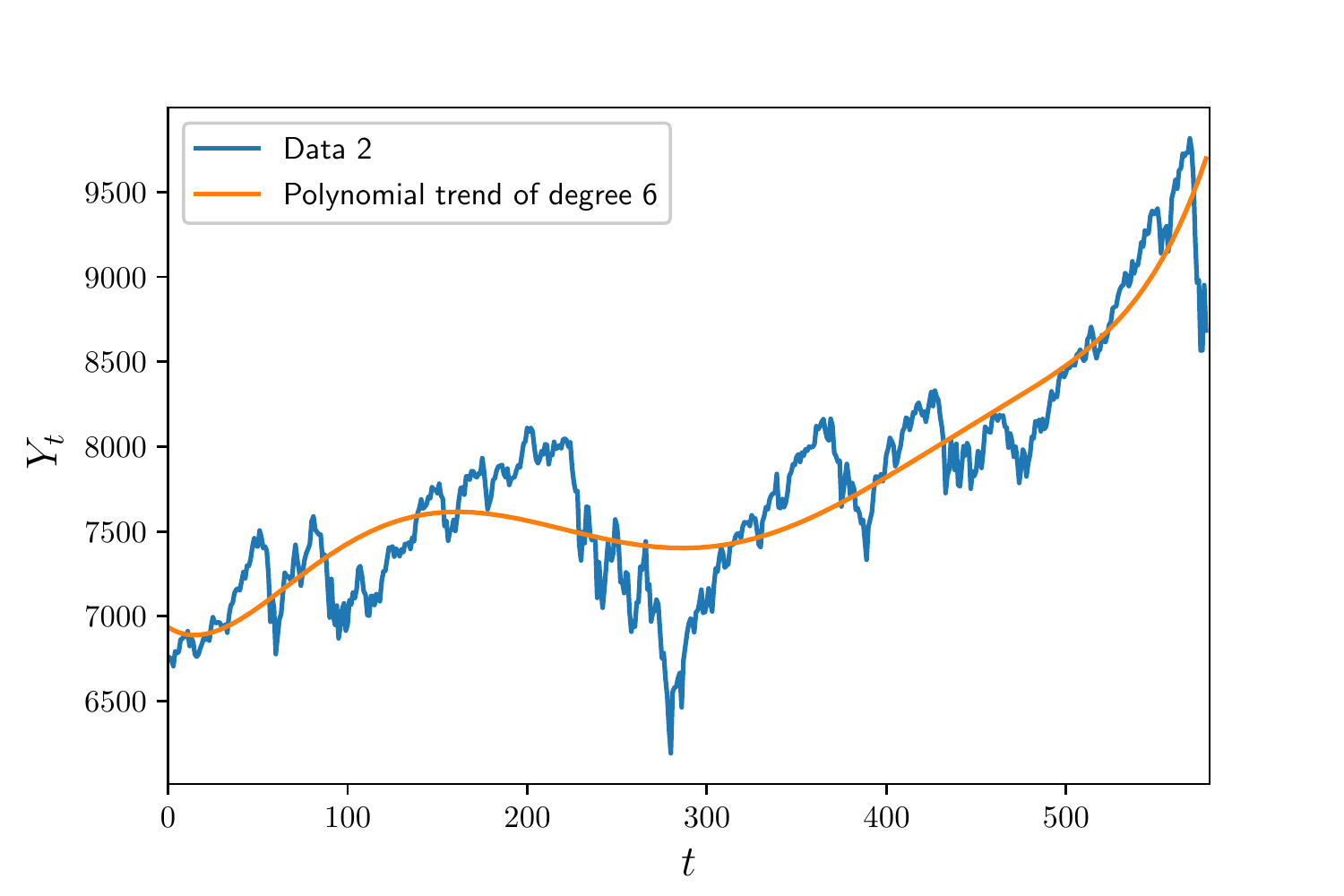}}
\caption{The  segmented data 1 (left panel) and data 2 (right panel) together with the fitted polynomials.}\label{fig6}
\end{figure}

The trends for both data sets were removed by using the degree 6 polynomial detrending. The trend was fitted by using the least squares method. The original data sets with the fitted polynomials are shown in Fig. \ref{fig6}. Next, for data 1 and data 2 we analyze the detrending time series and for each of them we use the partial autocorrelation function (PACF) to recognize the proper order of AR model. It is worth mentioning the PACF is a common tool to find the optimal order of the autoregressive models \cite{brockwell2016introduction}. We select the best order that is equal to the lag corresponding to the largest PACF value (except a lag equal to zero). We use the PACF plots to determine the components of AR($p$) model. Fig. \ref{fig7} shows the stationary data (after removing the trend) and  corresponding PACF plots indicating the optimal model - AR(1).
\begin{figure}[ht!]
\centering
\subfigure[Stationary data 1.]{
\includegraphics[width=8cm, height=5.5cm]{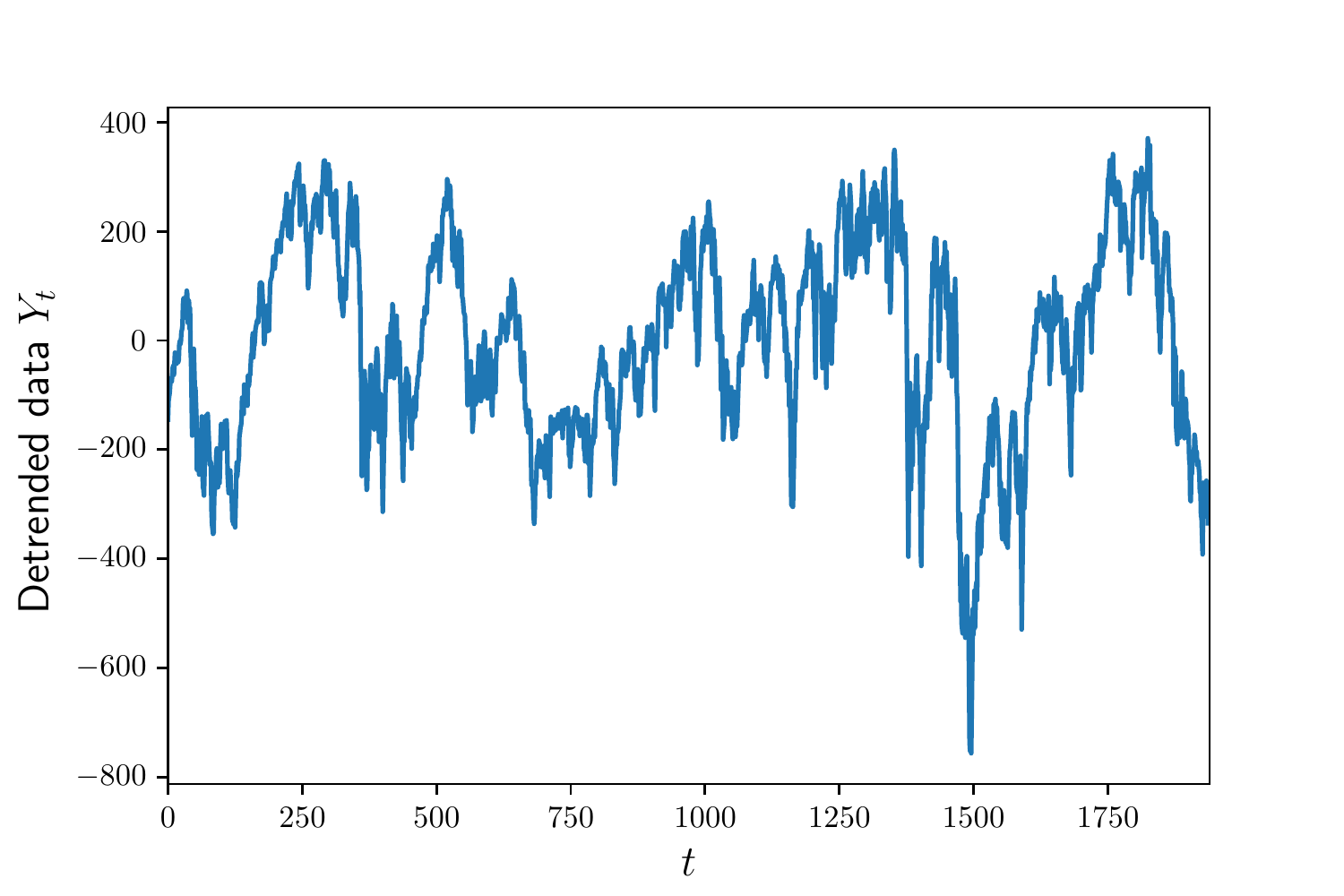}}
\subfigure[PACF of data 1.]{
\includegraphics[width=8cm, height=5.5cm]{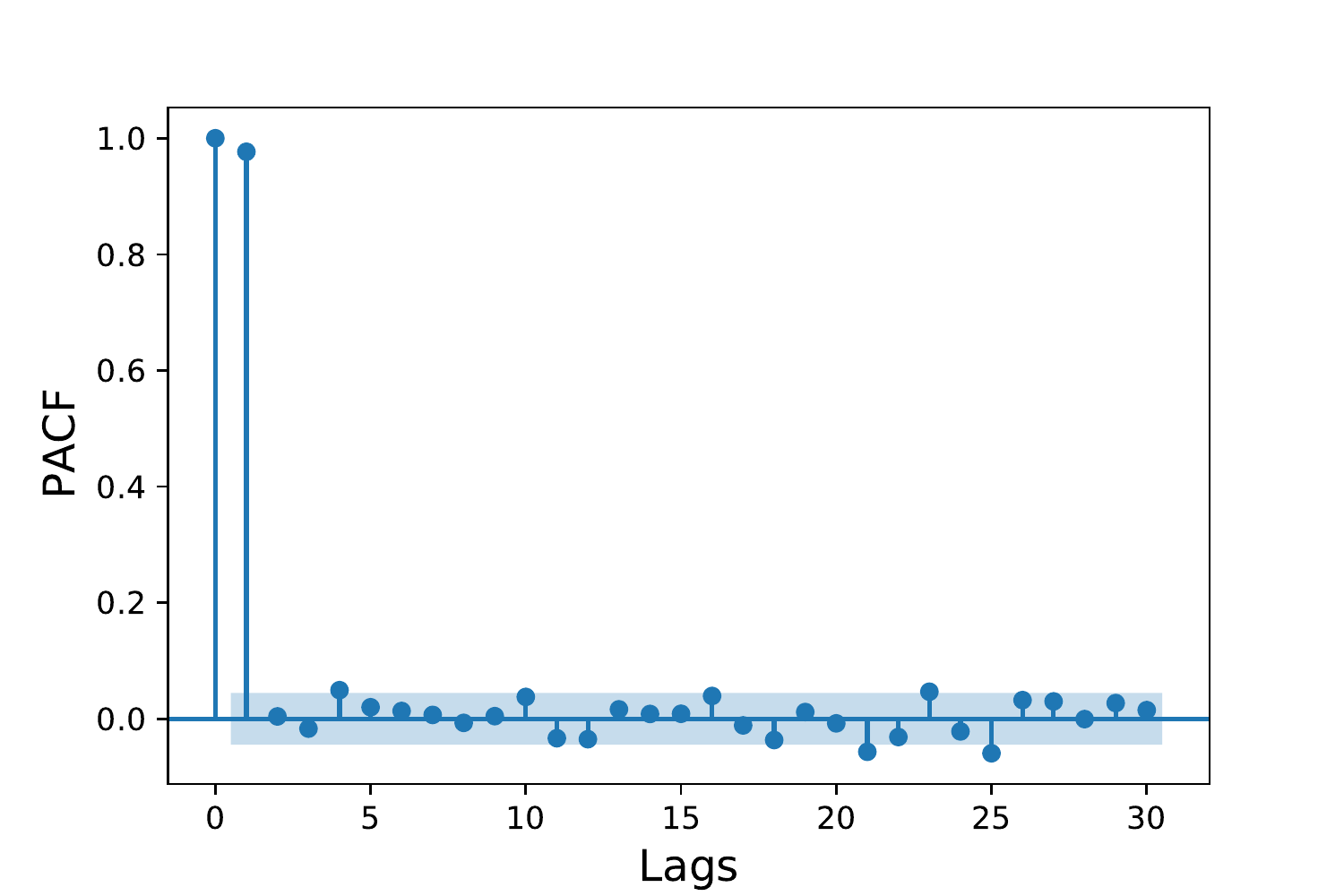}}

\subfigure[Stationary data 2.]{
\includegraphics[width=8cm, height=5.5cm]{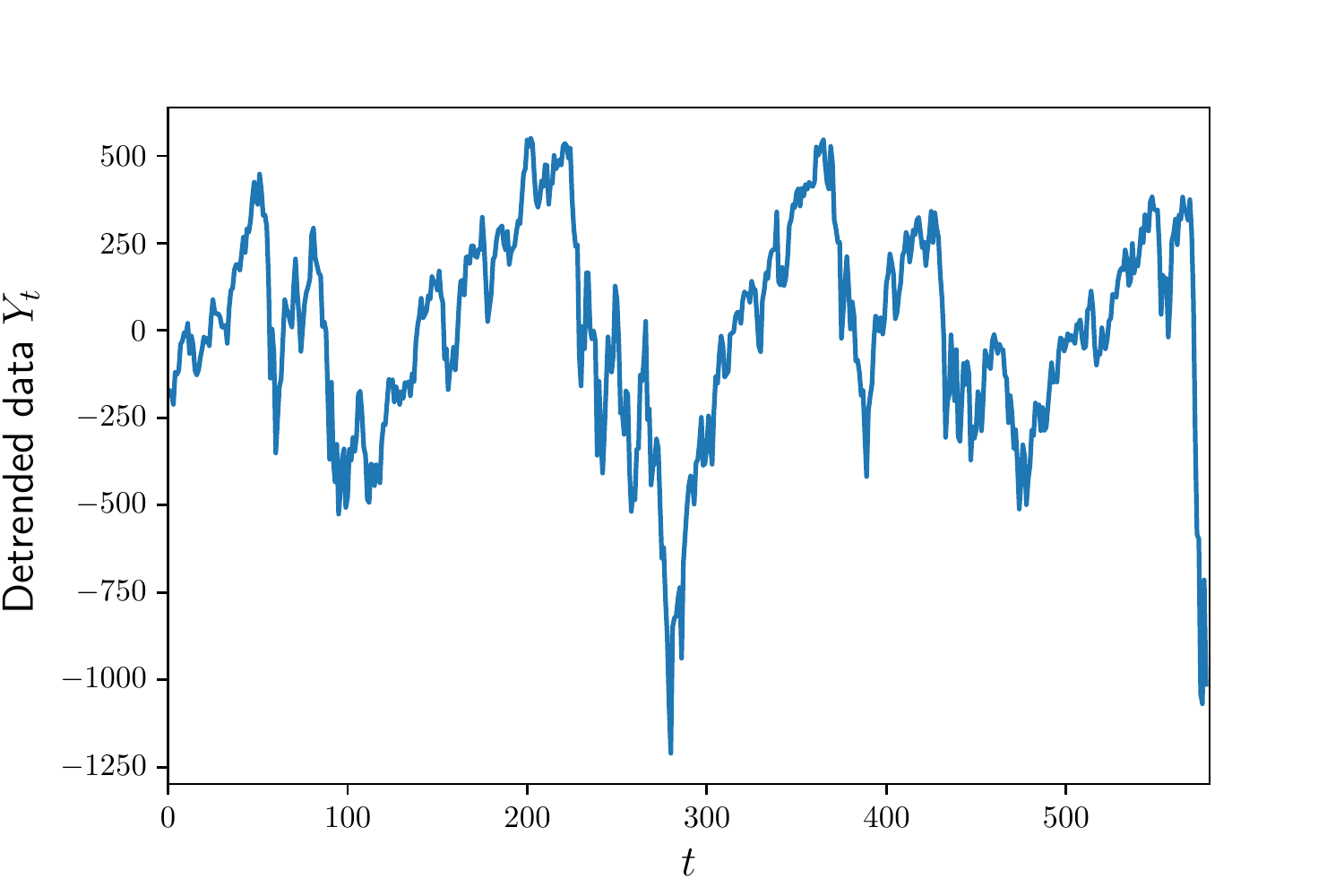}}
\subfigure[PACF of data 2.]{
\includegraphics[width=8cm, height=5.5cm]{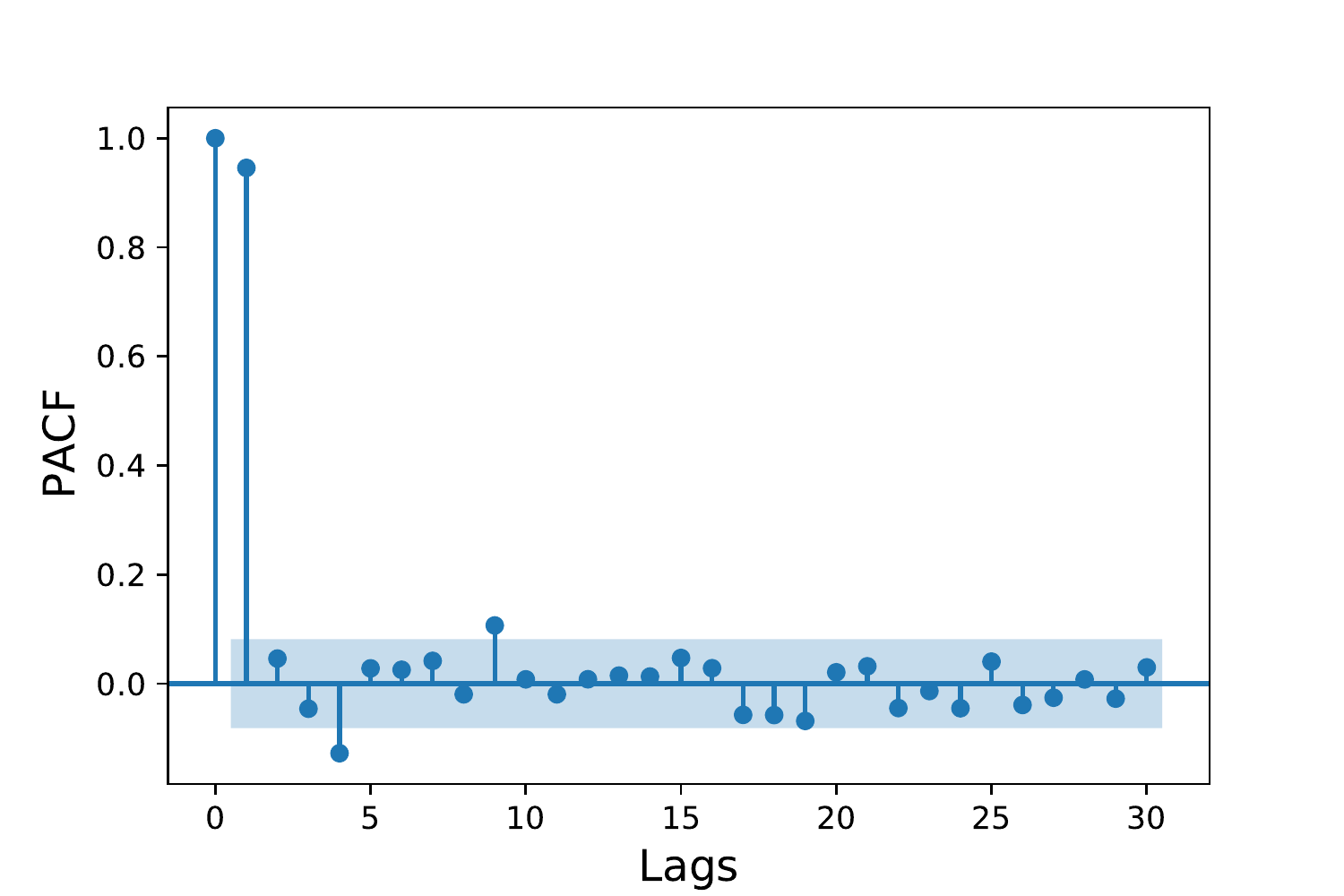}}
\caption{The time-series plot of stationary data 1 and data 2 (after removing the trend) - left panel and corresponding PACF plot - right panel.}\label{fig7}
\end{figure}
After above-described pre-processing steps, the EM algorithm is used to estimate the model's parameters. In the proposed model, $\mu = 0$ and $\beta =0$ parameters are kept fixed. The estimated values of parameters for data 1 and data 2 are summarised in Table \ref{tab1}.
\begin{center}
\begin{table}[ht!]
\centering
\begin{tabular}{||c||c||c||c|}
\hline 
 & $\hat{\rho}$ & $\hat{\delta}$  & $\hat{\alpha}$ \\ 
\hline 
Data 1 & $0.9809$ & $34.5837$ & $0.0226$\\ 
\hline 
Data 2 &  $0.9610$ & $70.3883$ & $0.0087$\\ 
\hline 
\end{tabular} 
\caption{Estimated parameters of data 1 and data 2 using EM algorithm.}\label{tab1}
\end{table}
\end{center}
As the final step, we analyze the innovation terms corresponding to data 1 and data 2 to confirm they can be modeled by using NIG distribution. The Kolmogorov-Smirnov (KS) test is used to check the normality of the residuals corresponding to data 1 and data 2. The KS test is a non-parametric test which is used as goodness-of-fit test by comparing the distribution of samples with the given probability distribution (one-sample KS test) or by comparing the empirical distribution function of two samples (2-sample KS test) \cite{ks1951}.  First, we use the one-sample KS test to reject the hypothesis of normal distribution of the residuals. The $p$-value of the KS test is $0$ for both cases, indicating that the null hypothesis (normal distribution) is rejected for both series. Thus, we applied two-sample KS test for both innovation terms with the null hypothesis of NIG distribution. The tested NIG distributions have the following parameters:  $\mu = 0,\; \beta = 0,\; \delta = 34.5\text{ and } \alpha = 0.02$ -  for residuals corresponding to data 1; and $\mu = 0,\; \beta = 0,\; \delta = 70.5 \text{ and } \alpha = 0.008$  - for the residuals corresponding to data 2. The $p$-values for 2-sample KS test are $0.565$ and $0.378$ for data 1 and data 2, respectively, which indicates that there is no evidence to reject the null hypothesis. Therefore, we assume that both the residual series follow the same distribution, implying that data 1 follows NIG$(\alpha=0.02,\; \beta=0,\; \mu=0,\; \delta=34.5)$ and data 2 has NIG$(\alpha=0.008,\; \beta=0,\; \mu=0,\; \delta=70.5)$. 
\begin{figure}[ht!]
\centering
\subfigure[QQ plot between innovations of data 1 and normal distribution.]{
\includegraphics[width=8cm, height=5.5cm]{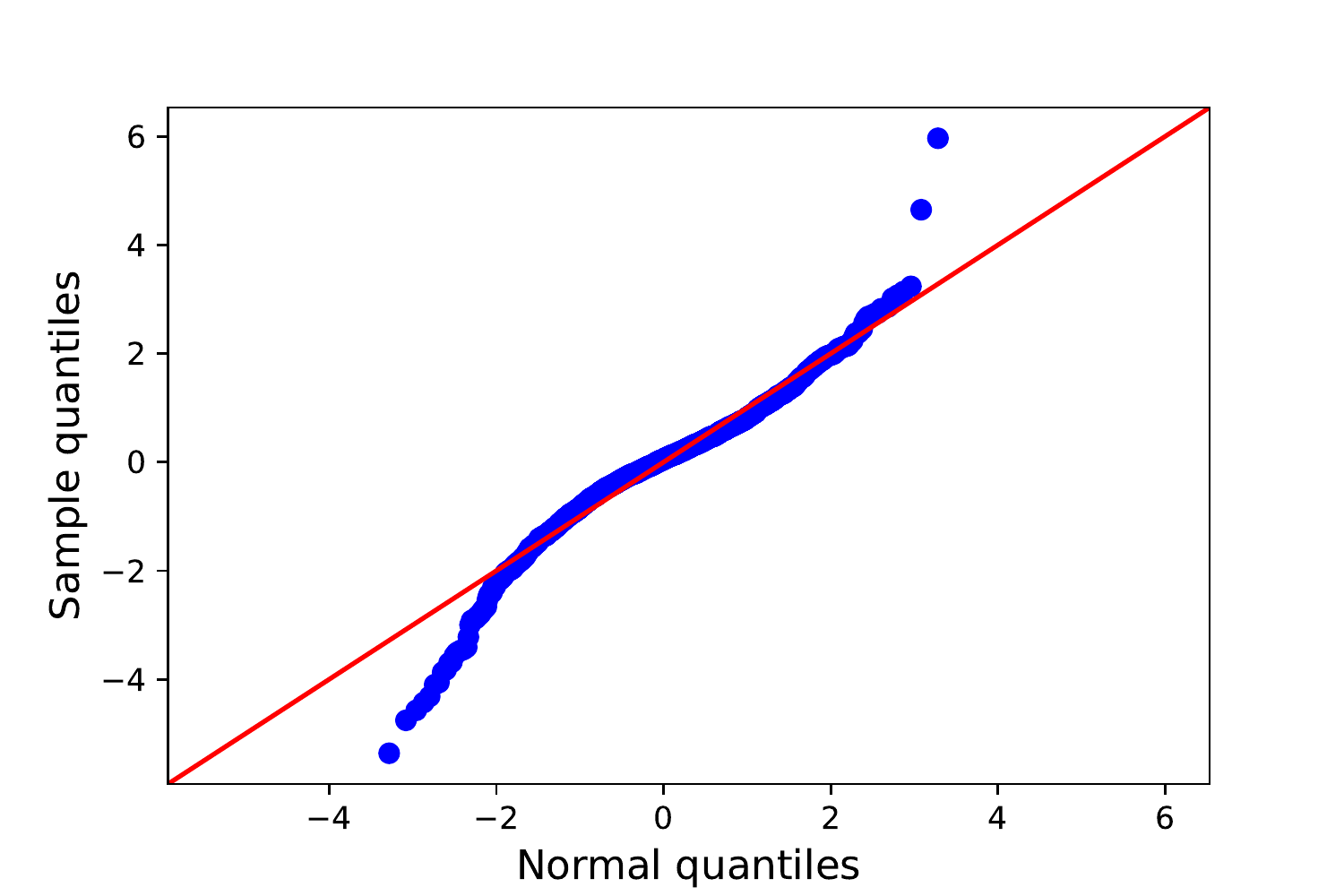}}
\subfigure[QQ plot between innovations of data 1 and NIG distribution.]{
\includegraphics[width=8cm, height=5.5cm]{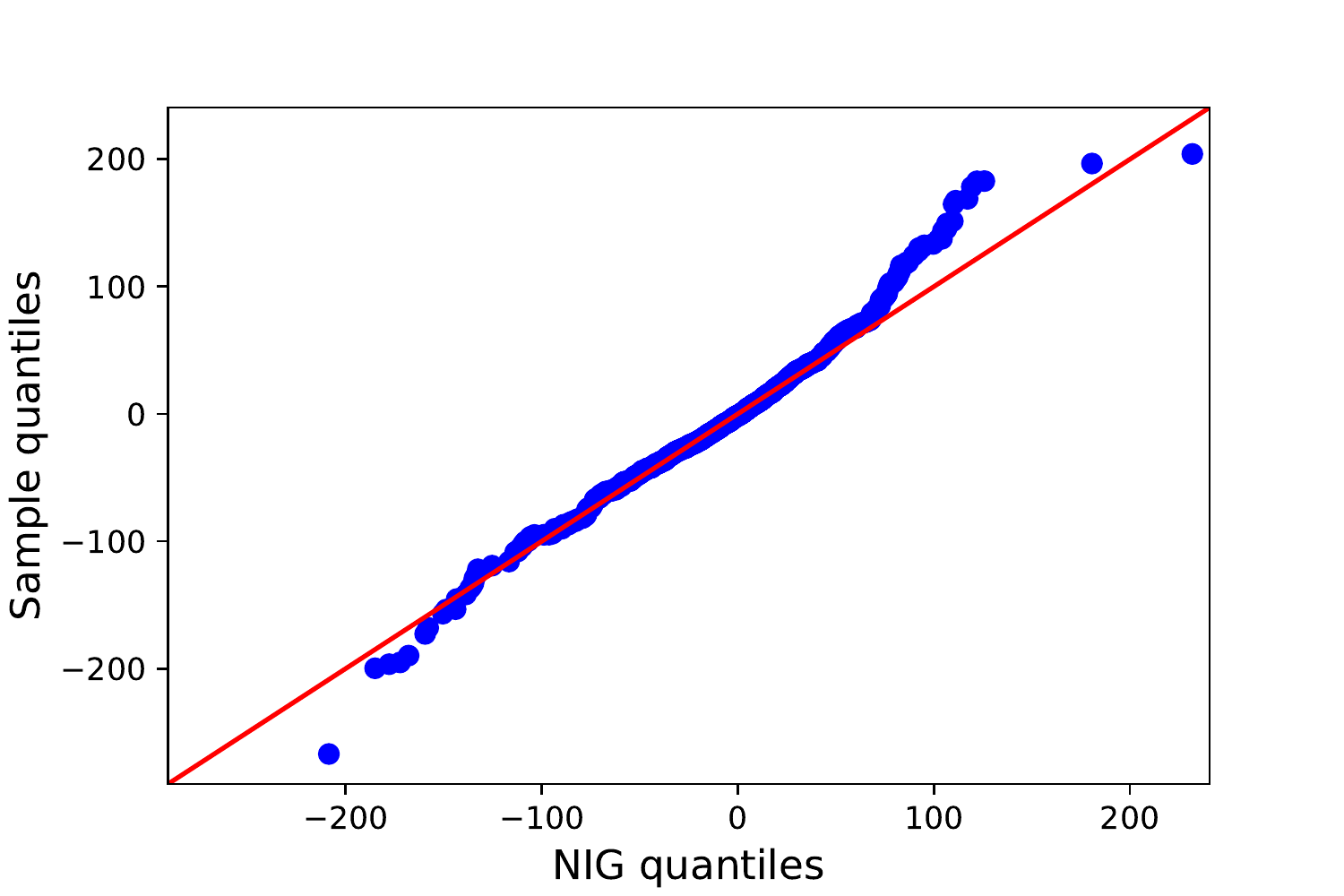}}
\caption{QQ plots of innovation terms of data 1 compared with (a) normal distribution and (b) NIG$(\alpha=0.02, \beta=0, \mu=0, \delta=34.5)$ distribution.}\label{fig8}
\end{figure} 
\begin{figure}[ht!]
\centering
\subfigure[QQ plot between innovations of data 2 and normal distribution.]{
\includegraphics[width=8cm, height=5.5cm]{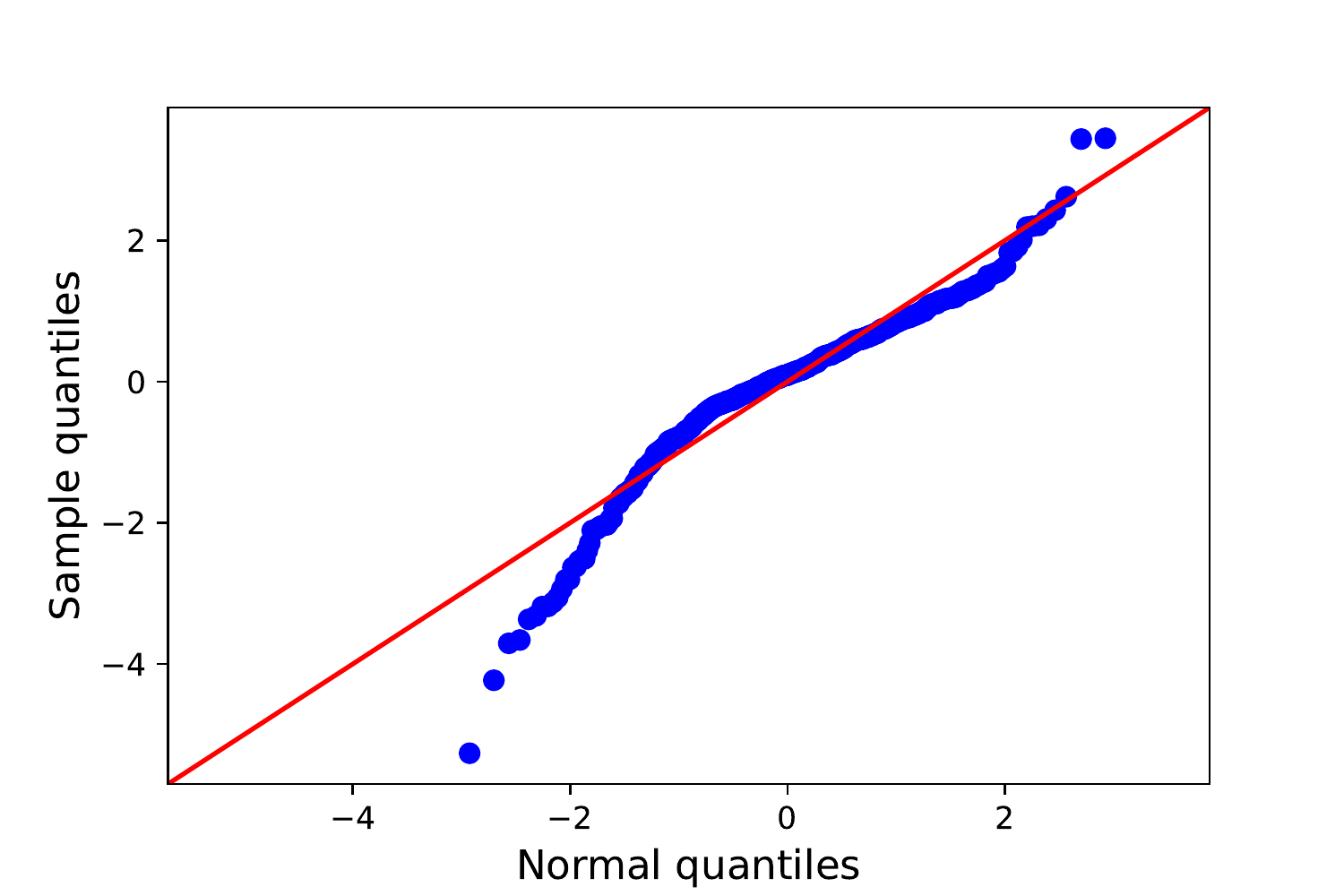}}
\subfigure[QQ plot between innovations of data 2 and NIG distribution.]{
\includegraphics[width=8cm, height=5.5cm]{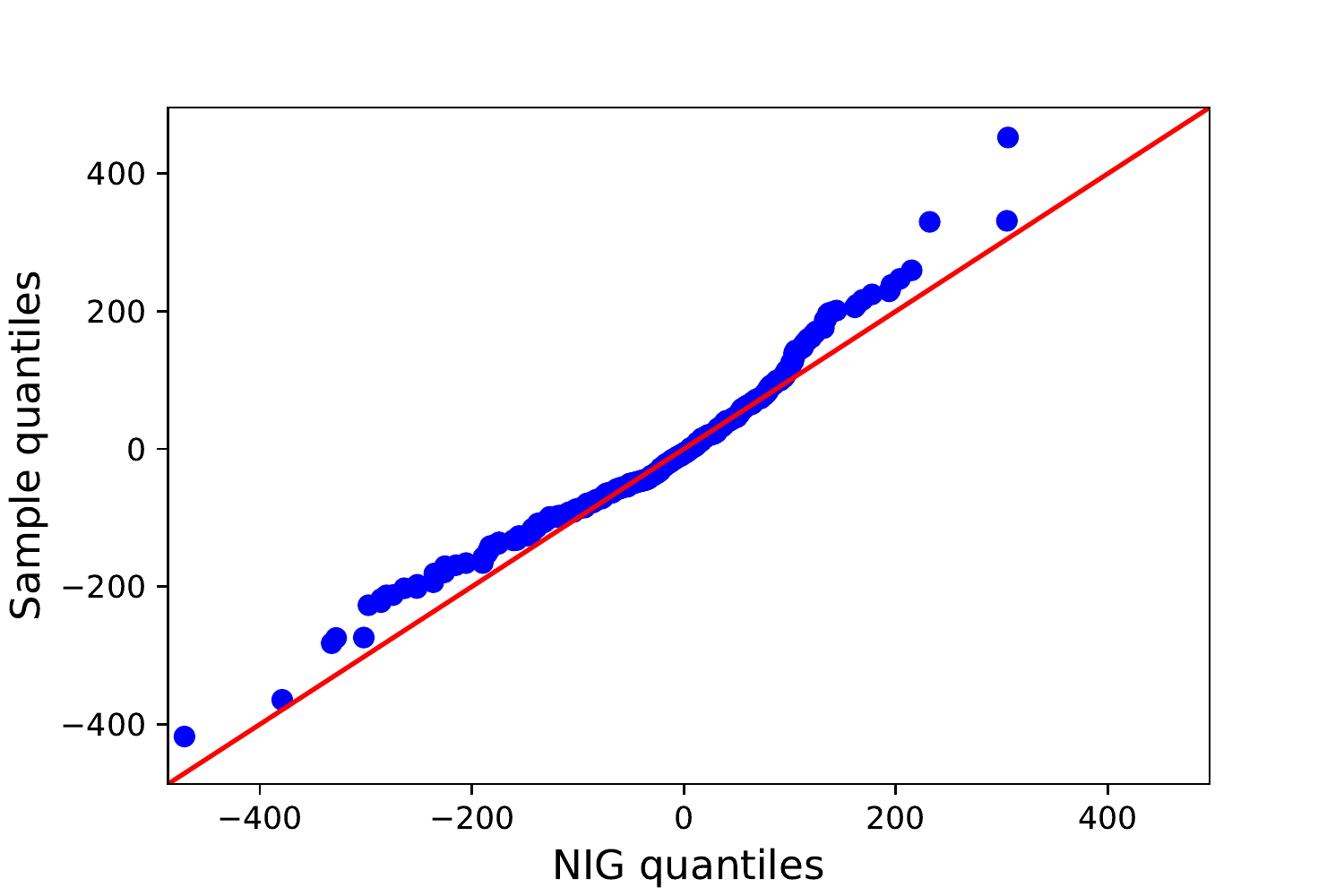}}
\caption{QQ plots of innovation terms of data 2 compared with (a) normal distribution and (b) NIG$(\alpha=0.008, \beta=0, \mu=0, \delta=70.5)$ distribution.}\label{fig9}
\end{figure}

To confirm that NIG distributions (with fitted parameters) are acceptable for the residual series, in Fig. \ref{fig8} and Fig. \ref{fig9} we demonstrate the QQ plot for the residuals of AR(1) models for data 1 and data 2 and the simulated data from normal (left panels) and corresponding NIG distributions (right panels). Observe that the tail of both data 1 and data 2 deviates from the red line on the left panels, which indicates that the data does not follow the  distribution. 

\begin{figure}[ht!]
\centering
\subfigure[KDE plot for data 1]{
\includegraphics[width=8cm, height=5.5cm]{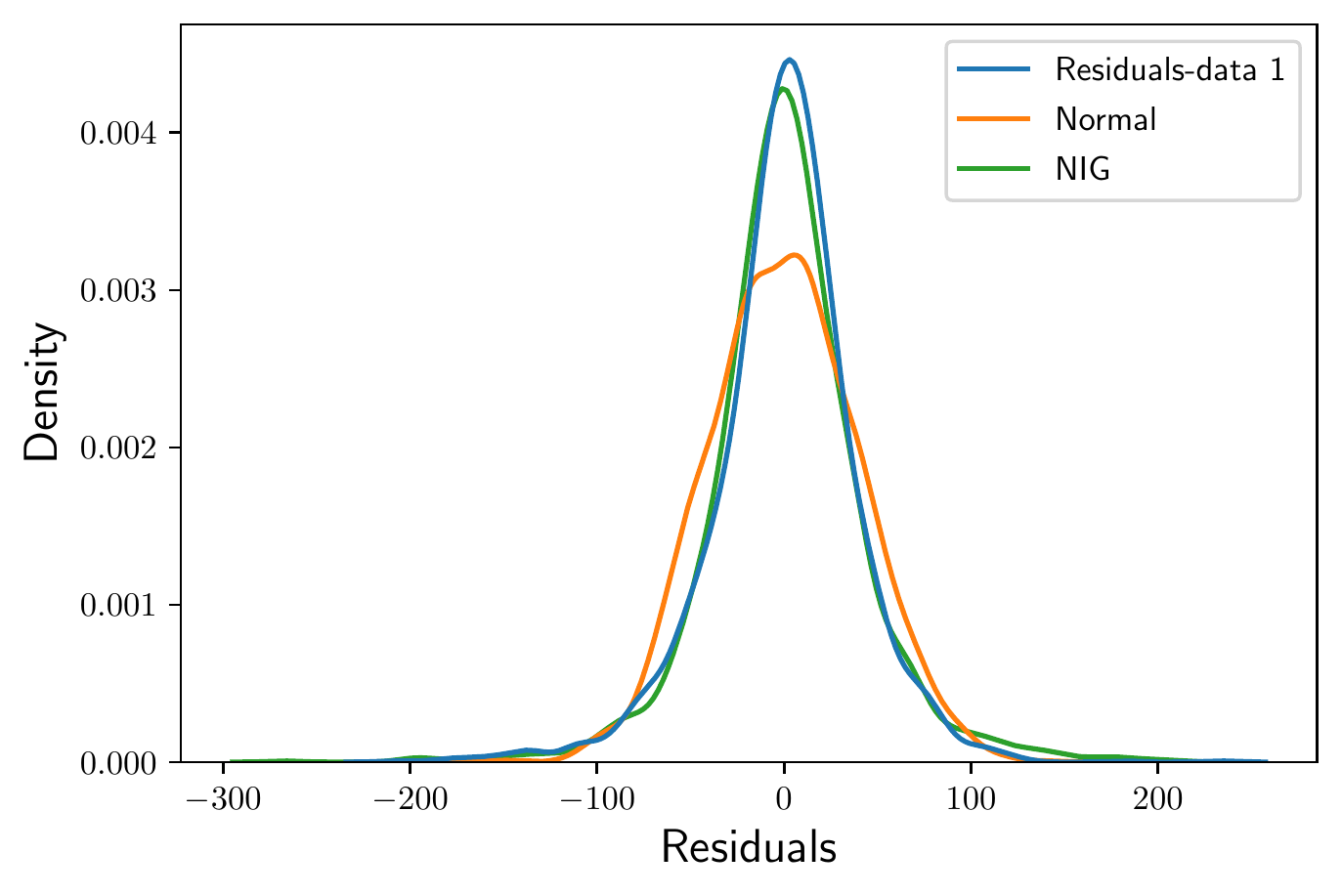}}\hfill
\subfigure[KDE plot for data 2]{
\includegraphics[width=8cm, height=5.5cm]{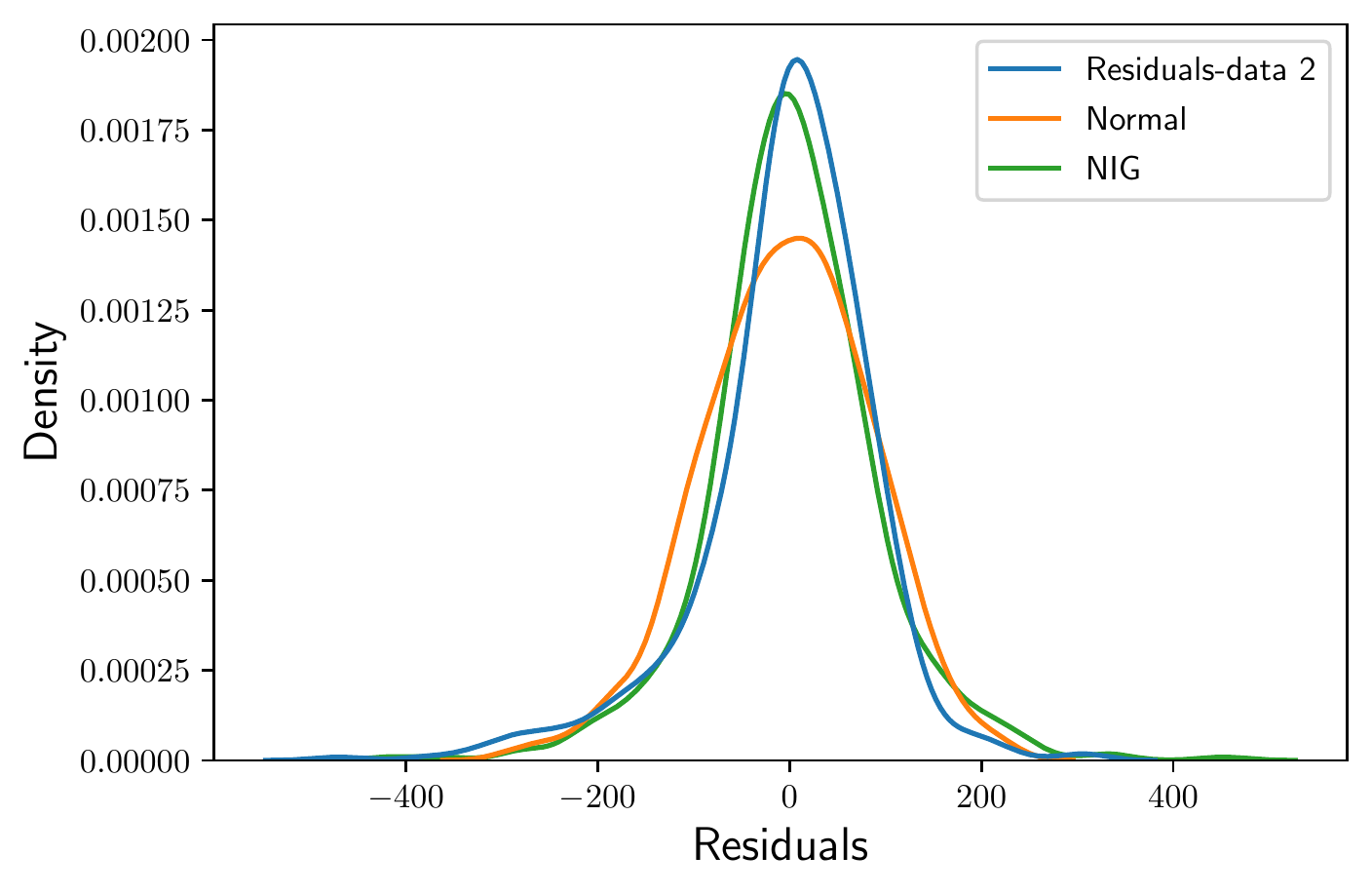}}
\caption{Kernel density estimation plots for comparing the innovation terms of data 1 with NIG$(\alpha=0.02, \beta=0, \mu=0, \delta=34.5)$ distribution and normal distribution $N(\mu=-0.1170, \sigma^2=1513.0754)$ (left panel)  and data 2 with NIG$(\alpha=0.008, \beta=0, \mu=0, \delta=70.5)$ and normal distribution $N(\mu=-1.6795, \sigma^2=89.2669)$ (right panel).}\label{kde}
\end{figure}

\begin{figure}[ht!]
\centering
\subfigure[Data 1]{
\includegraphics[width=8cm, height=6cm]{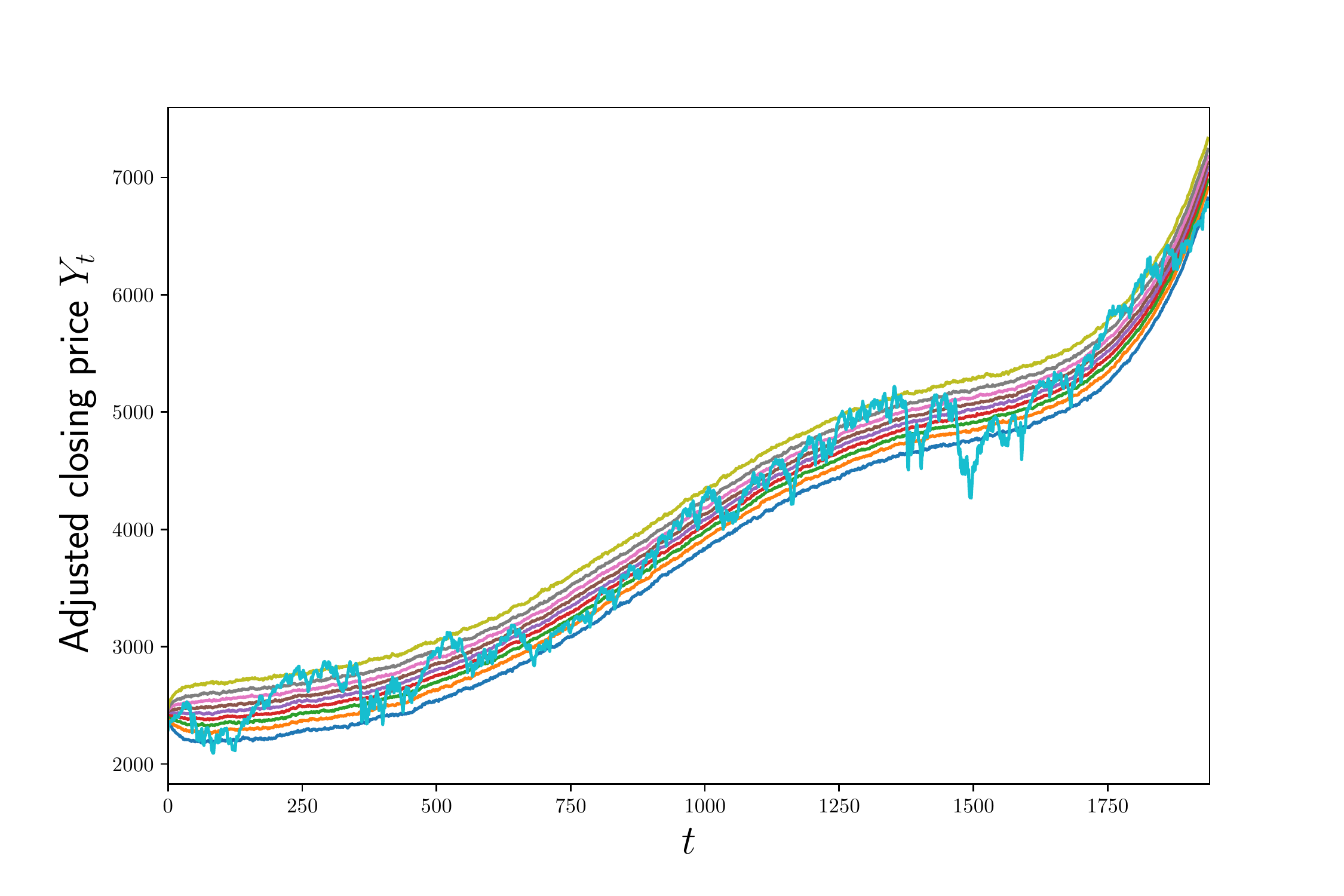}}
\subfigure[Data 2]{
\includegraphics[width=8cm, height=6cm]{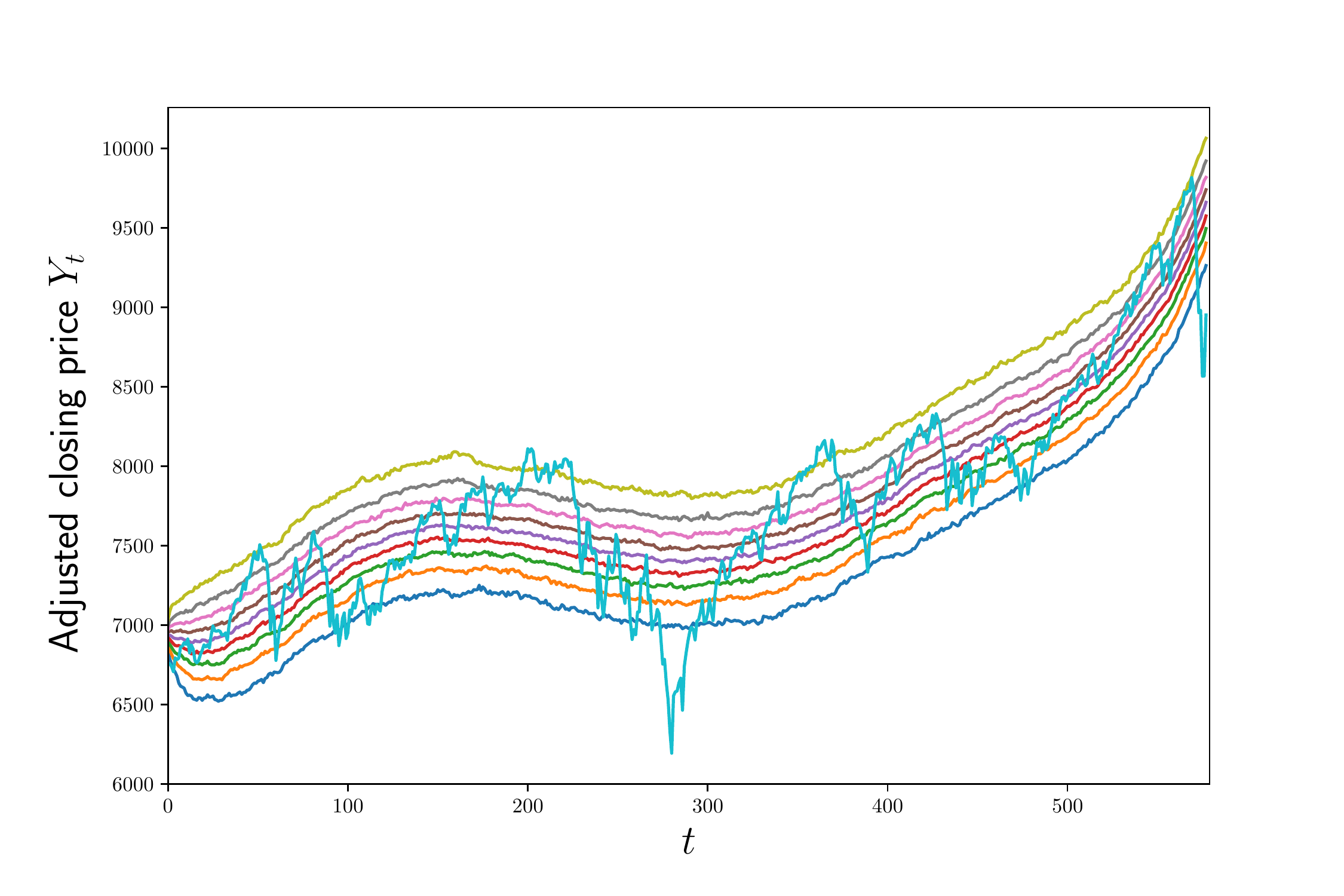}}
\caption{The adjusted closing price of NASDAQ index for both segments (blue lines) along with the quantile lines of $10\%, 20\%, ..., 90\%$ constructed base on the fitted AR(1) models with NIG distribution with added trends.}\label{fig11}
\end{figure}
The correspondence with NIG distribution is also demonstrated in Fig. \ref{kde}, where the kernel density estimation (KDE) plot is presented for the residual series and compared with the pdf of the normal and corresponding NIG distributions. A brief description of KDE method is given in \ref{App_E}. As one can see, the KDE plots clearly indicate the NIG distribution is the appropriate one for the residual series. 

Finally, to confirm that the fitted models are appropriate for data 1 and data 2, we constructed the quantile plots of the data simulated from the models for which the removed polynomials were added. The quantile lines are constructed based on $1000$ simulated trajectories with the same lengths as data 1 and data 2. In Fig. \ref{fig11} we present the constructed quantile lines on the levels  $10\%, 20\%,..., 90\%$  and the adjusted closing price of the NASDAQ index. 

The presented results for the real data indicate that the AR(1) model with innovation terms corresponding to NIG distribution can be useful for the financial data with visible extreme values.

\section{Conclusions} The heavy-tailed and semi-heavy-tailed distributions are at the heart of the financial time-series modeling. NIG is a semi-heavy tailed distribution which has tails heavier than the normal and lighter than the power law tails. In this article, the AR models with NIG distributed innovations are discussed. We have demonstrated the main properties of the analyzed systems. The main part is devoted to the new estimation algorithm for the considered models' parameters. The technique incorporates the EM algorithm widely used in the time series analysis. The effectiveness of the proposed algorithm is demonstrated for the simulated data. It is shown that the new technique outperforms the classical approaches based on the YW and CLS algorithms. Finally, we have demonstrated that an AR($1$) model with NIG innovations explain well the NASDAQ stock market index data for the period March 04, 2010 to March 03, 2020. We believe that the discussed model is universal and can be used to describe various real-life time series ranging from finance and economics to natural hazards, ecology, and environmental data. \\

\noindent{\bf Acknowledgements:}  Monika S. Dhull would like to thank the Ministry of Education (MoE), Government of India, for supporting her PhD research work. \\
The work of A.W. was supported by National Center of Science under Opus Grant No.
2020/37/B/HS4/00120 ``Market risk model identification and validation using novel statistical, probabilistic, and machine learning tools". A.K. would like to express his gratitude to Science and Engineering Research Board (SERB), India, for the financial support under the MATRICS
research grant MTR/2019/000286.


\medskip


\appendix

\section{Additional properties of NIG distribution}\label{App_A}
The following properties of NIG distribution are presented in \cite{Lillestol2000}.
Let $X\sim$ NIG $(\alpha, \beta, \mu, \delta)$, then following holds.
\begin{enumerate}[(a)]
    \item The moment generating function of $X$ is 
    $$
    M_X(u) = e^{\mu u + \delta(\sqrt{\alpha^2-\beta^2} - \sqrt{\alpha^2 - (\beta+u)^2}}.
    $$
    
    \item If $X\sim$ NIG $(\alpha, \beta, \mu, \delta)$, then $X+c \sim$ NIG $(\alpha, \beta, \mu + c, \delta)$, \; $c\in\mathbb{R}$
    
    \item If $X\sim$ NIG $(\alpha, \beta, \mu, \delta)$, then $cX \sim$ NIG $(\alpha/c, \beta/c, c\mu, c\delta)$, \;$c>0.$
    
    \item If $X_1\sim$ NIG $(\alpha, \beta, \mu_1, \delta_1)$ and $X_2\sim$ NIG $(\alpha, \beta, \mu_2, \delta_2)$ are independent then the sum $X_1+X_2\sim$ NIG $(\alpha, \beta, \mu_1+\mu_2, \delta_1+ \delta_2)$.
    
    \item If $X\sim$ NIG $(\alpha, \beta, \mu, \delta)$, then $\frac{X-\mu}{\delta}\sim$ NIG $(\alpha\delta, \beta\delta, 0, 1)$. 
    
\end{enumerate}

\section{Mean and variance of AR($p$) model with NIG innovations}\label{App_B}
We have $Y_t = \rho_1Y_{t-1} + \rho_2Y_{t-2}+\cdots+ \rho_pY_{t-p} + \epsilon_t$. Assuming stationarity, it follows that
$$
\mathbb{E}Y_t = \frac{\mathbb{E}(\epsilon_t)}{1-\rho_1-\rho_2-\cdots-\rho_p} = \frac{\mu+\delta\beta/\gamma}{1-\rho_1-\rho_2-\cdots-\rho_p}.
$$
Further, for $\mu = \beta =0$, we have $\mathbb{E}[Y_t] = \mathbb{E}[\epsilon_t] =0$ and 
\begin{align}\label{C1}
    Y_t^2 &= \rho_1Y_{t-1}Y_t + \rho_2Y_{t-2}Y_t+\cdots+ \rho_pY_{t-p}Y_t + \epsilon_tY_t.\nonumber\\
    \mathbb{E}[Y_t^2] &= \rho_1\mathbb{E}[Y_{t-1}Y_t] + \rho_2\mathbb{E}[Y_{t-2}Y_t]+\cdots+ \rho_p\mathbb{E}[Y_{t-p}Y_t] + \mathbb{E}[\epsilon_tY_t]\nonumber\\
    & = \rho_1\mathbb{E}[Y_{t-1}Y_t] + \rho_2\mathbb{E}[Y_{t-2}Y_t]+\cdots+ \rho_p\mathbb{E}[Y_{t-p}Y_t] + \mathbb{E}[\epsilon_t^2].\nonumber\\
    \rm{Var}[Y_t]& = \sigma_{\epsilon}^2 + \sum_{j=1}^{p}\rho_j\gamma_j,
\end{align}
where $\sigma_{\epsilon}^2 $ = Var$[\epsilon_t]$ and $\gamma_j = \mathbb{E}[Y_tY_{t-j}].$ Moreover, 
\begin{equation}\label{C2}
\gamma_j = \rho_1\gamma_{j-1} + \rho_2\gamma_{j-2}+\cdots+ \rho_p\gamma_{j-p},\;j\geq 1.
\end{equation}
For $p=2$, using \eqref{C1} and \eqref{C2}, it is easy to show that
$$
\rm{Var}[Y_t] = \frac{(1-\rho_2)\sigma_{\epsilon}^2}{1-\rho_2-\rho_1^2-\rho_2^2-\rho_1^2\rho_2 + \rho_2^3},
$$
where $\sigma_{\epsilon}^2 = \delta\alpha^2/\gamma^3.$ Again for $p=3$, usiigarg \eqref{C1} and \eqref{C2}, it follows
\footnotesize{
$$
\rm{Var}[Y_t] = \frac{(1-\rho_2-\rho_1\rho_3-\rho_3^2)\sigma_{\epsilon}^2}{1-\rho_2-\rho_1\rho_3-\rho_1^2-\rho_2^2-2\rho_3^2-\rho_1^2\rho_2-\rho_2^2\rho_3^2-\rho_1^2\rho_3^2-\rho_1^3\rho_3-4\rho_1\rho_2\rho_3+\rho_2\rho_3^2+\rho_1\rho_3^4+\rho_2^3+\rho_3^4+\rho_1\rho_2^2\rho_3}.
$$}

\section{Proof of Proposition 3.1}\label{App_C}

\begin{proof}
For AR(p) model, let $(\varepsilon_{t}, G_{t}),$ for $t = 1, 2, ..., n$ denote the complete data. The observed data $\varepsilon_{t}$ is assumed to be from NIG$(\alpha, \beta, \mu, \delta)$ and the unobserved data $G_{t}$ follows IG$(\gamma, \delta)$. We can write the innovation terms as follows
$$\varepsilon_{t} = Y_t - \bm{\rho^{T}}\mathbf{Y_{t-1}}, \text{ for } t = 1, 2, ..., n.$$
Note that $\varepsilon| G = g \sim N(\mu + \beta g, g)$ and the conditional pdf is
\begin{align*}
f(\varepsilon=\epsilon_t| G = g_t) 
= \frac{1}{\sqrt{2\pi g_t}}\exp\left(-\frac{1}{2g_t}(y_t - \bm{\rho^{T}}\mathbf{y_{t-1}} -\mu -\beta g_t)^2\right). 
\end{align*} 
Now, we need to estimate the unknown parameters ${\theta} = (\alpha, \beta, \mu, \delta, \bm{\rho}^T)$. We find the conditional expectation of log-likelihood of unobserved/complete data $(\varepsilon, G)$ with respect to the conditional distribution of $G$ given $\varepsilon$. Since the unobserved data is assumed to be from IG$(\gamma, \delta)$ therefore, the posterior distribution is generalised inverse Gaussian (GIG) distribution i.e.,
$$G | \varepsilon, {\theta}\quad \sim  \text{  GIG}( -1, \delta \sqrt{\phi(\epsilon)}, \alpha).$$
The conditional first moment and inverse first moment are as follows:
\begin{align*}
\mathbb{E}(G|\epsilon) &= \frac{\delta \phi(\epsilon)^{1/2}}{\alpha}\frac{K_0(\alpha \delta \phi(\epsilon)^{1/2})}{K_1(\alpha \delta \phi(\epsilon)^{1/2})},\\
\mathbb{E}(G^{-1}|\epsilon) &= \frac{\alpha}{\delta \phi(\epsilon)^{1/2}}\frac{K_{-2}(\alpha \delta \phi(\epsilon)^{1/2})}{K_{-1}(\alpha \delta \phi(\epsilon)^{1/2})}.
\end{align*}

\noindent These first order moments will be used in calculating the conditional expectation of the log-likelihood function.
The complete data likelihood is given by
\begin{align*}
L(\theta) &= \prod_{t=1}^{n}f(\epsilon_{t}, g_t) = \prod_{t=1}^{n}f_{\varepsilon|G}(\epsilon_{t} | g_t) f_{G}(g_t)\\
&= \prod_{t=1}^{n} \frac{\delta}{2 \pi g_{t}^2}\exp(\delta \gamma) \exp\left(-\frac{\delta^2}{2g_t} - \frac{\gamma^2 g_t}{2} - \frac{g_t^{-1}}{2}(\epsilon_{t} - \mu)^2 - \frac{\beta^{2} g_t}{2} + \beta(\epsilon_t - \mu)\right).
\end{align*}  

\noindent The log likelihood function will be
\begin{align*}
l(\theta) &= n\log(\delta) - n\log(2 \pi) + n\delta \gamma - n\beta \mu -2\sum_{t=1}^{n}\log(g_t) - \frac{\delta^2}{2} \sum_{t=1}^{n} g_t^{-1} \\
&- \frac{\gamma^2}{2} \sum_{t=1}^{n} g_t  - \frac{1}{2}\sum_{t=1}^{n} g_t^{-1}(\epsilon_t -\mu)^2 - \frac{\beta^2}{2} \sum_{t=1}^{n} g_t + \beta \sum_{t=1}^{n}\epsilon_t.
\end{align*}
Now in {\it E-step} of EM algorithm, we need to compute the expected value of complete data log likelihood known as $Q(\theta|\theta^k)$, which is expressed as 

\begin{align*}
Q(\theta|\theta^{(k)}) &= \mathbb{E}_{G|\varepsilon,\theta^{(k)}}[\log f(\varepsilon,G|\theta)|\varepsilon, \theta^{(k)}]= \mathbb{E}_{G|\varepsilon,\theta^{(k)}}[L(\theta|\theta^{(k)})]\\
&= n\log \delta + n\delta \gamma - n\beta \mu - n\log(2 \pi) - 2\sum_{t=1}^{n}\mathbb{E}(\log g_t| \epsilon_t, \theta^{(k)}) - \frac{\delta^2}{2}\sum_{t=1}^{n} w_t\\
&- \frac{\gamma^2}{2}\sum_{t=1}^{n} s_t - \frac{\beta^2}{2}\sum_{t=1}^{n} s_t + \beta \sum_{t=1}^n(y_t - \bm{\rho^{T}}\mathbf{Y_{t-1}}) - \frac{1}{2}\sum_{t=1}^{n}(y_t - \bm{\rho^{T}}\mathbf{Y_{t-1}} - \mu)^2 w_t,
\end{align*}
where, $s_t = \mathbb{E}_{G|\varepsilon,\theta^{(k)}}(g_t|\epsilon_t, \theta^{(k)})$ and 
$w_t = \mathbb{E}_{G|\varepsilon,\theta^{(k)}}(g_t^{-1}|\epsilon_t, \theta^{(k)})$. 
Update the parameters by maximizing the $Q$ function using the following equations
\begin{align*}
    \frac{\partial{Q}}{\partial{\bm{\rho}}} &= \sum_{t=1}^{n}w_{t}(y_{t}-\bm{\rho^{T}}  
\mathbf{Y_{t-1}}-\mu)\bm{Y_{t-1}^{T}} - \beta \sum_{t=1}^{n}\bm{Y_{t-1}^{T}},\\
    \frac{\partial{Q}}{\partial{\mu}} &= -n\beta + \sum_{t=1}^{n}w_{t}(\epsilon_{t} - \mu),\\
    \frac{\partial{Q}}{\partial{\beta}} &= -n\mu + \sum_{t=1}^{n}y_{t} - \beta \sum_{t=1}^{n}s_{t} - \sum_{t=1}^{n}\bm{\rho^{T}}\mathbf{Y_{t-1}}, \\
    \frac{\partial{Q}}{\partial{\delta}} &= n\gamma + \frac{n}{\delta} - \delta \sum_{t=1}^{n}w_{t}, \\
    \frac{\partial{Q}}{\partial{\gamma}} &= n\delta - \gamma\sum_{t=1}^{n}s_{t}.
\end{align*}

\noindent Solving the above equations, we obtain the following estimates of the parameters

\begin{align*}
\hat{\bm{\rho}} &= \left(\displaystyle\sum_{t=1}^{n} w_{t}\bm{Y_{t-1}Y_{t-1}^{T}} \right)^{-1} \displaystyle\sum_{t=1}^{n}(w_{t} y_{t} - \mu w_{t} - \beta)Y_{t-1}
\end{align*}

\begin{align}\label{main_EM}
\begin{split}
\hat {\mu} &= \frac{-n\beta + \displaystyle\sum_{t=1}^{n}\epsilon_{t}w_{t}}{n\bar{w}_t};\\
\hat {\beta} &= \frac{\displaystyle\sum_{t=1}^{n}w_{t} \epsilon_{t}-n\bar{w}_t\bar{\epsilon}_t}{n(1-\bar{s}_t\bar{w}_t)};\\
\hat {\delta} &= \sqrt{\frac{\bar{s}}{(\bar{s} \bar{w} - 1)}},\;\; \hat {\gamma} = \frac{\delta}{\bar{s}},\;\mbox{and}\;
\hat{\alpha} = (\gamma^2 + \beta^2)^{1/2},
\end{split}
\end{align} 

where $\bar{s} = \frac{1}{n}{\sum_{t=1}^{n}s_t}, \bar{w} = \frac{1}{n}{\sum_{t=1}^{n}w_t}. $
\end{proof}

\section{KDE method} \label{App_E}
The KDE method also known as Parzen–Rosenblatt window method, is a non-parametric approach to find the underlying probability distribution of data. It is a technique that lets one to create a smooth curve given a set of data and one of the most famous method for density estimation. For a sample $S = \{x_i\}_{i=1,2,...,N}$ having distribution function $f(x)$ has the kernel density estimate $\widehat{f}(x)$  defined as \cite{elgammal2002}
$$\widehat{f}(x) = \frac{1}{N}\sum_{n=1}^{N}K_{\sigma}(x-x_i), $$
where $K_{\sigma}$ is kernel function with bandwidth $\sigma$ such that $K_{\sigma}(t)=(\frac{1}{\sigma})K(\frac{t}{\sigma})$.

\end{document}